%% file: main.tex
\documentclass[11pt]{article}

\usepackage{hyperref}
\usepackage{fullpage}
\usepackage{amsmath,amsfonts,amsthm,amssymb,xspace,bm, verbatim,dsfont,mathtools}
\usepackage{graphicx}
\usepackage{url,algorithm, algorithmic}
\usepackage{color}
\usepackage{epstopdf}
\usepackage[titletoc,title]{appendix}
\usepackage{comment}

\theoremstyle{plain}
\newtheorem{theorem}{Theorem}[section]
\newtheorem{lemma}[theorem]{Lemma}
\newtheorem*{lemma*}{Lemma}

\newtheorem{corollary}[theorem]{Corollary}
\theoremstyle{definition}

\newcommand{\rn}{\mathbb{R}^n}
\newcommand{\rnn}{\mathbb{R}^{n \times n}}

\newcommand{\vecfont}[1]{\mathbf{#1}}

\newcommand{\e}{\vecfont{e}}

\newcommand{\eps}{\epsilon}

\newcommand{\Var}{\operatorname{Var}}

\newcommand{\note}[1]{\marginpar{\tiny *note in TeX*}}
\newcommand{\ignore}[1]{}

\renewcommand{\phi}{\varphi}

\newcommand{\R}{\mathbb{R}}

\newcommand{\E}{\mathbb{E}}
\newcommand{\expec}[1]{\mathbb{E}\left[#1\right]}
\newcommand{\prob}[1]{\mathbb{P}\left[#1\right]}

\newcommand{\Vt}{V^{(t+1)}}
\newcommand{\Vht}{\widehat{V}^{(t+1)}}
\newcommand{\Ut}{U^{(t)}}
\newcommand{\Uht}{\widehat{U}^{(t)}}
\newcommand{\Rt}{R^{(t+1)}}
\newcommand{\Vo}{V^*}
\newcommand{\So}{\Sigma^*}
\newcommand{\Uo}{U^*}

\newcommand{\so}{\sigma^*}
\newcommand{\vto}{v^{t+1}}
\newcommand{\ut}{u^t}

\newcommand{\uo}{u^* }

\newcommand{\vo}{v^* }

\newcommand{\wvto}{\widehat{v}^{t+1}}

\newcommand{\ip}[2]{\langle #1, #2 \rangle}
\newcommand{\vt}{\vto}
\newcommand{\vht}{\wvto}

\def\Ro{R_{\Omega}}

\def\lv{\left\vert}
\def\rv{\right\vert}
\def\lV{\left\Vert}
\def\rV{\right\Vert}
\def\dij{\delta_{ij}}
\def\wij{w_{ij}}

\def\qij{q_{ij}}
\def\qhij{\hat{q}_{ij}}

\newcommand{\uht}{\widehat{u}^{t}}

\def\Mk{M_{\mathbf{r_k}}}
\def\Uk{U_{\mathbf{r_k}}}

\def\Rok{R_{\Omega^k}}
\begin{document}
\title{Tighter Low-rank Approximation via Sampling the Leveraged Element}
\author{
Srinadh Bhojanapalli\\
{The University of Texas at Austin}\\
{bsrinadh@utexas.edu}
\and
Prateek Jain\\
{Microsoft Research, India}\\
{prajain@microsoft.com}
\and
Sujay Sanghavi\\
{The University of Texas at Austin}\\
{sanghavi@mail.utexas.edu}
}
\maketitle

\begin{abstract}

\input{abstract}

\end{abstract}
\newpage

\input{intro}
\input{related}
\input{lrapprox}
\input{coherent_noisy}
\input{covariance}
\input{distributedPCA}
\input{simulations}



\appendix
\input{appendix}

\input{appendixb}
\input{appendixc}
\end{document}

%% file: abstract.tex
In this work, we propose a new randomized algorithm for computing a low-rank approximation to a given matrix. Taking an approach different from existing literature, our method first involves a specific biased sampling, with an element being chosen based on the leverage scores of its row and column, and then involves weighted alternating minimization over the factored form of the intended low-rank matrix, to minimize error only on these samples. Our method can leverage input sparsity, yet produce approximations in {\em spectral} (as opposed to the weaker Frobenius) norm; this combines the best aspects of otherwise disparate current results, but with a dependence on the condition number $\kappa = \sigma_1/\sigma_r$. In particular we require $O(nnz(M) + \frac{n\kappa^2 r^5}{\epsilon^2} )$ computations to generate a rank-$r$ approximation to $M$ in spectral norm. In contrast, the best existing method requires $O(nnz(M)+ \frac{nr^2}{\epsilon^4})$ time to compute an approximation in Frobenius norm. Besides the tightness in spectral norm, we have a better dependence on the error $\epsilon$. Our method is naturally and highly parallelizable.

Our new approach enables two extensions that are interesting on their own. The first is a new method to directly compute a low-rank approximation (in efficient factored form) to the product of two given matrices; it computes a small random set of entries of the product, and then executes weighted alternating minimization (as before) on these. The sampling strategy is different because now we cannot access leverage scores of the product matrix (but instead have to work with input matrices). The second extension is an improved algorithm with smaller communication complexity for the distributed PCA setting (where each server has small set of rows of the matrix, and want to compute low rank approximation with small amount of communication with other servers). 

%% file: intro.tex
\section{Introduction}\label{sec:intro}

Finding a low-rank approximation to a matrix is fundamental to a wide array of machine learning techniques. The large sizes of modern data matrices has driven much recent work into efficient (typically randomized) methods to find low-rank approximations that do not exactly minimize the residual, but run much faster / parallel, with fewer passes over the data. Existing approaches involve either intelligent sampling of a few rows / columns of the matrix, projections onto lower-dimensional spaces, or sampling of entries followed by a top-$r$ SVD of the resulting matrix (with unsampled entries set to 0). 

We pursue a different approach: we first sample entries in a specific biased random way, and then minimize the error on these samples over a search space that is the factored form of the low-rank matrix we are trying to find. We note that this is different from approximating a 0-filled matrix; it is instead reminiscent of matrix completion in the sense that it only looks at errors on the sampled entries. Another crucial ingredient is how the sampling is done: we use a combination of $\ell_1$ sampling, and of a distribution where the probability of an element is proportional to the sum of the leverage scores of its row and its column. 

Both the sampling and the subsequent alternating minimization are naturally fast, parallelizable, and able to utilize sparsity in the input matrix. Existing literature has either focused on running in input sparsity time but approximation in (the weaker) Frobenius norm, or running in $O(n^2)$ time with approximation in spectral norm. Our method provides the best of both worlds: it runs in input sparsity time, with just two passes over the data matrix, and yields an approximation in spectral norm. It does however have a dependence on the ratio of the first to the $r^{th}$ singular value of the matrix. 

Our alternative approach also yields new methods for two related problems: directly finding the low-rank approximation of the product of two given matrices, and distributed PCA. 

%
%
%

{\bf Our contributions} are thus three new methods in this space:
\begin{itemize}
\item {\bf Low-rank approximation of a general matrix:} Our first (and main) contribution is a new method (LELA, Algorithm 1) for low-rank approximation of any given matrix: first draw a random subset of entries in a specific biased way, and then execute a weighted alternating minimization algorithm that minimizes the error on these samples over a factored form of the intended low-rank matrix. The sampling is done with only two passes over the matrix (each in input sparsity time), and both the sampling and the alternating minimization steps are highly parallelizable and compactly stored/manipulated. 

For a matrix $M$, let $M_r$ be the best rank-$r$ approximation (i.e. the matrix corresponding to  top $r$ components of SVD). Our algorithm finds a rank-$r$ matrix $\widehat{M}_r$ in time $O(nnz(M) + \frac{n\kappa^2 r^5}{\epsilon^2} )$, while providing approximation in spectral norm: $\|M-\widehat{M}_r\| \leq \|M-M_r\| +\epsilon \|M-M_r\|_F$, where $\kappa=\sigma_1(M)/\sigma_r(M)$ is the condition number of $M_r$. Existing methods either can run in input sparsity time, but provide approximations in (the weaker) Frobenius norm (i.e. with $||\cdot ||$ replaced by $||\cdot||_F$ in the above expression), or run in $O(n^2)$ time to approximate in spectral norm, but even then with leading constants larger than 1. Our method however does have a dependence on $\kappa$, which these do not. See Table~\ref{table:1} for a detailed comparison to  existing results for low-rank approximation.

\item {\bf Direct approximation of a matrix product:} We provide a new method to directly find a low-rank approximation to the product of two matrices, without having to first compute the product itself. To do so, we first choose a small set of entries (in a biased random way) of the product that we will compute, and then again run weighted alternating minimization on these samples. The choice of the biased random distribution is now different from above, as the sampling step does not have access to the product matrix. However, again both the sampling and alternating minimization are highly parallelizable.

For $A\in \R^{n_1\times d}$, $B\in \R^{d\times n_2}$, and $n=max(n_1,n_2)$, our algorithm first chooses $O(nr^3 \log n / \epsilon^2)$ entries of the product $A\cdot B$ that it needs to sample; each sample takes $O(d)$ time individually, since it is a product of two length-$d$ vectors (though these can be parallelized). The weighted alternating minimization then runs in $O(\frac{nr^5 \kappa^2}{\eps^2})$ time (where $\kappa=\sigma_1(A\cdot B)/\sigma_r(A\cdot B)$). This results in a rank-$r$ approximation $\widehat{AB}_r$ of $A\cdot B$  in spectral norm, as given above.

\item {\bf Distributed PCA:} Motivated by applications with really large matrices, recent work has looked at low-rank approximation in a distributed setting where there are $s$ servers -- each have small set of rows of the matrix -- each of which can communicate with a central processor charged with coordinating the algorithm. In this model, one is interested in find good approximations while minimizing both computations and the communication burden on the center. 

We show that our LELA algorithm can be extended to the distributed setting while guaranteeing small communication complexity. In particular, our algorithm guarantees the same error bounds as that of our non-distributed version but guarantees communication complexity of $O(ds+ \frac{n r^5\kappa^2}{\eps^2}\log n)$ real numbers for computing rank-$r$ approximation to $M\in \R^{n\times d}$. For $n\approx d$ and large $s$, our analysis guarantees significantly lesser communication complexity than the state-of-the-art method \cite{kannan2014principal}, while providing tighter spectral norm bounds. 
\end{itemize}

%% file: related.tex
{\bf Notation}: 
Capital letter $M$ typically denotes a matrix. $M^i$ denotes the $i$-th row of $M$, $M_j$ denotes the $j$-th column of $M$, and $M_{ij}$ denotes the $(i,j)$-th element of $M$. Unless specified otherwise, $M\in \R^{n\times d}$ and $M_r$ is the best rank-$r$ approximation of $M$. Also, $M_r=\Uo\Sigma^* (\Vo)^T$ denotes the SVD of $M_r$. $\kappa=\sigma_1^*/\sigma_r^*$ denotes the condition number of $M_r$, where $\sigma_i^*$ is the $i$-th singular value of $M$. $\|u\|$ denotes the $L_2$ norm of vector $u$. $\|M\|=\max_{\|x\|=1}\|Mx\|$ denotes the spectral or operator norm of $M$. $\|M\|_F=\sqrt{\sum_{ij}M_{ij}^2}$ denotes the Frobenius norm of $M$. Also, $\|M\|_{1,1}=\sum_{ij}|M_{ij}|$. $dist(X, Y)=\|X_\perp^TY\|$ denotes the principal angle based distance between subspaces spanned by $X$ and $Y$ orthonormal matrices. Typically, $C$ denotes a global constant independent of problem parameters and can change from step to step. 

Given a set $\Omega\subseteq [n]\times [d]$, $P_{\Omega}(M)$ is given by: $P_{\Omega}(M)(i,j)=M_{ij}$ if $(i,j)\in \Omega$ and $0$ otherwise. $R_{\Omega}(M)=w.*P_{\Omega}(M)$ denotes the Hadamard product of $w$ and $P_{\Omega}(M)$. That is, $R_{\Omega}(M)(i,j)=w_{ij}M_{ij}$ if $(i,j)\in \Omega$ and $0$ otherwise. Similarly let $R_{\Omega}^{1/2} (M)(i, j) =\sqrt{w_{ij}}M_{ij}$ if $(i,j)\in \Omega$ and $0$ otherwise.
\section{Related results}\label{sec:related}

\begin{table*}
\centering
    \begin{tabular}{| l | l | l | l |}
    \hline
   Reference & Frobenius norm  & Spectral norm  & Computation time \\ \hline
   {\bf BJS14 (Our Algorithm)} &  $ (1+\eps)\|\Delta\|_F$ &  $\|\Delta\| + \eps \|\Delta\|_F $ & $ O(nnz(M) + \frac{nr^5\kappa^2\log(n)}{\eps^2}) $ \\ \hline
   CW13\cite{clarkson2013low} & $(1+\eps)\|\Delta\|_F $ & $(1+\eps)\|\Delta\|_F$ & $O(nnz(M) + \frac{nr^2}{\eps^4}+\frac{r^3}{\eps^5}) $ \\ \hline
  BG13 \cite{ boutsidis2013improved} & $(1+\eps)\|\Delta\|_F $ & $c\|\Delta\|+ \eps \|\Delta\|_F$ & $O(n^2(\frac{r +\log(n)}{\eps^2})+ n\frac{r^2 \log(n)^2}{\eps^4})$ \\ \hline
   NDT09\cite{ nguyen2009fast} & $(1+\eps)\|\Delta\|_F $ & $c \|\Delta\|+ \eps\sqrt{n} \|\Delta\|$ & $ O(n^2 \log(\frac{r\log(n)}{\eps}) +\frac{nr^2\log(n)^2}{\eps^4} )$\\ \hline
   WLRT08\cite{ woolfe2008fast} & $(1+\eps)\|\Delta\|_F $ & $\|\Delta\|+\eps \sqrt{n}\|\Delta\|$ & $ O(n^2 \log(\frac{r}{\eps}) +\frac{nr^4}{\eps^4})$ \\ \hline
   Sar06\cite{ sarlos2006improved} & $(1+\eps)\|\Delta\|_F $ & $(1+\eps)\|\Delta\|_F$ & $ O(nnz(M) \frac{r}{\eps} +n \frac{r^2}{\epsilon^2})$ \\ \hline
    \end{tabular}
\caption{Comparison of error rates and computation time of some low rank approximation algorithms. $\Delta=M-M_r$.}
\label{table:1}
\end{table*}

\noindent {\bf Low rank approximation:} Now we will briefly review some of the existing work on algorithms for low rank approximation. \cite{frieze1998fast} introduced the problem of computing low rank approximation of a matrix $M$ with few passes over $M$. They presented an algorithm that samples few rows and columns and does SVD to compute low rank approximation, and gave additive error guarantees. \cite{drineas2003pass, drineas2006fast} have extended these results. \cite{achlioptas2001fast} considered a different approach based on entrywise sampling and quantization for low rank approximation and has given additive error bounds. 

\cite{har2006low, sarlos2006improved, drineas2006subspace, deshpande2006adaptive} have given low rank approximation algorithms with relative error guarantees in Frobenius norm. \cite{woolfe2008fast, nguyen2009fast} have provided guarantees on error in spectral norm which are later improved in~\cite{halko2011finding, boutsidis2013improved}. The main techniques of these algorithms is to use a random Gaussian or Hadamard transform matrix for projecting the matrix onto a low dimensional subspace and compute the rank-$r$ subspace. \cite{boutsidis2013improved} have given an algorithm based on random Hadamard transform that computes rank-$r$ approximation in time $O(\frac{n^2r}{\eps^2})$ and gives spectral norm bound of $\|M-\widehat{M}_r\| \leq c\|M-M_r\| + \eps \|M-M_r\|_F$.

One drawback of Hadamard transform is that it cannot take advantage of sparsity of the input matrix. Recently \cite{clarkson2013low} gave an algorithm using sparse subspace embedding that runs in input sparsity time with relative Frobenius norm error guarantees.

We presented some results in this area as a comparison with our results in table~\ref{table:1}. This is a heavily subsampled set of existing results on low rank approximations. There is a lot of interesting work on very related problems of computing column/row based(CUR) decompositions, matrix sketching, low rank approximation with streaming data. Look at \cite{mahoney2011randomized, halko2011finding} for more detailed discussion and comparison.\\

\noindent{\bf Matrix sparsification:} In the matrix sparsification problem, the goal is to create a sparse sketch of a given matrix by sampling and reweighing the entries of the matrix. Various techniques for sampling have been proposed and analyzed which guarantee $\epsilon$ approximation error in Frobenius norm with $O(\frac{n}{\eps^2}\log n)$ samples~ \cite{drineas2011note, achlioptas2013near}. As we will see in the next section, the first step of algorithm~\ref{alg:1} involves sampling according to a very specific distribution (similar to matrix sparsification), which has been designed for guaranteeing good error bounds for computing low rank approximation. For a comparison of various sampling distributions for the problem of low rank matrix recovery see~\cite{chen2014coherent}.\\

\noindent {\bf Matrix completion:} Matrix completion problem is to recover a $n \times n$ rank-$r$ matrix from observing small number of ($O(nr\log(n))$) random entries. Nuclear norm minimization is shown to recover the matrix from uniform random samples if the matrix is incoherent\footnote{A $n \times d$ matrix $A$ of rank-$r$ with SVD $\Uo \So (\Vo)^T$ is incoherent if $\|(\Uo)^i\|^2 \leq \frac{\mu_0 r}{n}, \forall i$ and $\|(\Vo)^j\|^2 \leq \frac{\mu_0 r}{d}, \forall j$ for some constant $\mu_0$.} \cite{candes2009exact, candes2010power, recht2009simpler, gross2011recovering} . Similar results are shown for other algorithms like OptSpace~\cite{keshavan2010matrix} and alternating minimization~\cite{jain2013low, hardt2013understanding, hardt2014fast}. Recently \cite{chen2014coherent} has given guarantees for recovery of any matrix under leverage score sampling from $O(nr\log^2(n))$ entries.\\

\noindent {\bf Distributed PCA:} In distributed PCA, one wants to compute rank-$r$ approximation of a $n \times d$ matrix that is stored across $s$ servers with small communication between servers. One popular model is row partition model where subset of rows are stored at each server. Algorithms in \cite{feldman2013turning, liberty2013simple, ghashami2014relative, liang2013distributed} achieve $O(\frac{dsr}{\eps})$ communication complexity with relative error guarantees in Frobenius norm, under this model. Recently \cite{kannan2014principal} have considered the scenario of arbitrary splitting of a $n \times d$ matrix and given an algorithm that has $O(\frac{dsr}{\eps})$ communication complexity with relative error guarantees in Frobenius norm.


%% file: lrapprox.tex
\newcommand{\argmin}{\arg\!\min}
\section{Low-rank Approximation of Matrices}\label{sec:lela}

In this section we will present our main contribution: a new randomized algorithm for computing low-rank approximation of any given matrix. Our algorithm first samples a few elements from the given matrix $M\in \R^{n\times d}$, and then rank-$r$ approximation is computed using only those samples. Algorithm~\ref{alg:1} provides a detailed pseudo-code of our algorithm; we now comment on each of the two stages: \\


{\em Sampling:} A crucial ingredient of our approach is using the correct sampling distribution. Recent results in matrix completion \cite{chen2014coherent} indicate that a small number ($O(nr\log^2(n))$) of samples drawn in a way biased by leverage scores\footnote{If SVD of $M_r =\Uo \So (\Vo)^T$ then leverage scores of $M_r$ are $||(\Uo)^i||^2$ and $||(\Vo)^j||^2$ for all $i, j$.} can capture {\em all} the information in any exactly low-rank matrix. While this is indicative, here we have neither access to the leverage scores, nor is our matrix exactly low-rank. We approximate the leverage scores with the row and column norms ($||M^i||^2$ and $||M_j||^2$), and account for the arbitrary high-rank nature of input by including an $L_1$ term in the sampling; the distribution is given in eq. (\ref{eq:prob}). Computationally, our sampling procedure can be done in two passes and $O( nnz(M)+m\log n)$ time.  \\

{\em Weighted alternating minimization:} In our second step, we directly optimize over the factored form of the intended low-rank matrix, by minimizing a {\em weighted} squared error over the sampled elements from stage 1. That is, we first express the low-rank approximation $\widehat{M}_r$ as $UV^T$ and then iterate over $U$ and $V$ alternatingly to minimize the weighted $L_2$ error over the {\em sampled} entries (see Sub-routine~\ref{algo:2}). Note that this is different from taking principal components of a 0-filled version of the sampled matrix. The weights give higher emphasis to elements with smaller sampling probabilities. In particular, the goal is to minimize the following objective function:
\begin{equation}
  \label{eq:obj_am}
  Err(\widehat{M}_r)=\sum_{(i,j)\in \Omega}w_{ij} \left(M_{ij}-(\widehat{M}_r)_{ij}\right)^2,
\end{equation}
where $w_{ij}=1/\qhij$ when $\qhij >0$, $0$ else. For initialization of the WAltMin procedure, we compute SVD of $R_{\Omega}(M)$ (reweighed sampled matrix) followed by a trimming step (see Step 4, 5 of Sub-routine~\ref{algo:2}). Trimming step sets $(\tilde{U}^0)^i=0$ if $\|({U}^0)^i\|\geq 4\|M^i\|/\|M\|_F$ and $(\tilde{U}^0)^i=(U^0)^i$ otherwise; and  $\widehat{U}^0$ is the orthonormal matrix spanning the column space of $\tilde{U}^0$. This step prevents heavy rows/columns from having undue influence.

We now provide our main result for low-rank approximation and show that Algorithm~\ref{alg:1} can provide a tight approximation to $M_r$ while using a small number of samples $m=\E[|\Omega|]$. 
\begin{theorem}
\label{thm:main}
Let $M\in \R^{n\times d}$ be any given matrix ($n\geq d$) and let $M_r$ be the best rank-$r$ approximation to $M$. Set the number of samples $m=\frac{C}{\gamma} \frac{n r^3}{\epsilon^2} \kappa^2 \log(n) \log^2(\frac{\|M\|}{\zeta})$, where $C>0$ is any global constant, $\kappa=\sigma_1/\sigma_r$ where $\sigma_i$ is the $i$-th singular value of $M$. Also, set the number of iterations of WAltMin procedure to be $T= \log(\frac{\|M\|}{\zeta})$. Then, with probability greater than $1-\gamma$ for any constant $\gamma>0$, the output $\widehat{M}_r$ of Algorithm~\ref{alg:1} with the above specified parameters $m, T$, satisfies:  \[\|M-\widehat{M}_r\| \leq   \|M-M_r\|+ \epsilon  \lV M-M_r \rV_F+\zeta.\]
That is, if $T=\log(\frac{\|M\|}{\epsilon \|M-M_r\|_F})$, we have: 
\[\|M-\widehat{M}_r\| \leq   \|M-M_r\|+ 2\epsilon  \lV M-M_r \rV_F.\]
\end{theorem}
Note that our time and sample complexity depends quadratically on $\kappa$. Recent results in the matrix completion literature shows that such a dependence can be improved to $\log(\kappa)$ by using a slightly more involved analysis \cite{hardt2014fast}. We  believe a similar analysis can be combined with our techniques to obtain tighter bounds; we leave a similar tighter analysis for future research as such a proof would be significantly more tedious and would take away from the key message of this paper. 
\begin{algorithm}[t]
\caption{LELA: Leveraged Element Low-rank Approximation}
\label{alg:1}
\begin{algorithmic}[1]
\INPUT matrix: $M\in \R^{n\times d}$, rank: $r$, number of samples: $m$, number of iterations: $T$
\STATE Sample $\Omega\subseteq [n]\times [d]$ where each element is sampled independently with probability: $\qhij = \min \{ \qij, 1 \}$ 
\begin{equation}
  \label{eq:prob}
  q_{ij}=m\cdot \left(\frac{\|M^i\|^2+\|M_j\|^2}{2(n+d) \|M\|_F^2}+\frac{|M_{ij}|}{2\|M\|_{1,1}}\right).
\end{equation}
/*See Section~\ref{sec:lela_comp} for details about efficient implementation of this step*/
\STATE Obtain $P_{\Omega}(M)$ using one pass over $M$
\STATE $\widehat{M}_r={\sf WAltMin}(P_{\Omega}(M),  \Omega, r, \hat{q}, T)$
\OUTPUT $\widehat{M}_r$
\end{algorithmic}
\end{algorithm}

\floatname{algorithm}{Sub-routine}

\begin{algorithm}[t]
\caption{WAltMin: Weighted Alternating Minimization}
\label{algo:2}
\begin{algorithmic}[1]
\INPUT $P_{\Omega}(M),\  \Omega,\ r, \ \hat{q},\ T$
\STATE $w_{ij}=1/\qhij$ when $\qhij >0$, $0$ else, $\forall i, j$
\STATE Divide $\Omega$ in $2T+1$ equal uniformly random subsets, i.e., $\Omega=\{\Omega_0, \dots, \Omega_{2T}\}$
\STATE $R_{\Omega_0}(M)\gets w.*P_{\Omega_0}(M)$
\STATE $U^{(0)} \Sigma^{(0)} (V^{(0)})^T=SVD(R_{\Omega_0}(M), r)$ //Best rank-$r$ approximation of $R_{\Omega_0}(M)$
\STATE Trim $U^{(0)}$ and let $\widehat{U}^{(0)}$ be the output (see Section~\ref{sec:lela})
\FOR {$t=0$ to $T-1$}
	\STATE $\Vht= \argmin_{V }\| R_{\Omega_{2t+1}}^{1/2}(M- \widehat{U}^{(t)} V^T)\|_F^2$, for $V \in \R^{d \times r}$.
	\STATE $\widehat{U}^{(t+1)}=\argmin_{U}\| R_{\Omega_{2t+2}}^{1/2}(M-U(\widehat{V}^{(t+1)})^T)\|_F^2$ , for $U \in \R^{n \times r}$.
\ENDFOR
\OUTPUT Completed matrix $\widehat{M}_r=\widehat{U}^{(T)} (\widehat{V}^{(T)})^T$. 
\end{algorithmic}
\end{algorithm}
\floatname{algorithm}{Algorithm}
\subsection{Computation complexity:}\label{sec:lela_comp}
In the first step we take 2 passes over the matrix to compute the sampling distribution~\eqref{eq:prob} and sampling the entries based on this distribution. It is easy to show that this step would require $O(nnz(M) +m\log(n))$ time. Next, the initialization step of WAltMin procedure requires computing rank-$r$ SVD of $R_{\Omega_0}(M)$ which has at most $m$ non-zero entries. Hence, the procedure can be completed in $O(mr)$ time using standard techniques like power method. Note that by Lemma~\ref{lem:approx_init} we need top-$r$ singular vectors of $R_{\Omega_0}(M)$ only upto constant approximation. Further $t$-th iteration of alternating minimization takes $O(2|\Omega_{2t+1}| r^2)$ time. So, the total time complexity of our method is $O(nnz(M)+mr^2)$. As shown in Theorem~\ref{thm:main}, our method requires $m=O(\frac{nr^3}{\eps^2}\kappa^2 \log(n) \log^2(\frac{\|M\|}{\eps\|M-M_r\|_F}))$ samples. Hence, the total run-time of our algorithm is: $O(nnz(M) +\frac{nr^5}{\eps^2}\kappa^2 \log(n) \log^2(\frac{\|M\|}{\eps\|M-M_r\|_F})) $.\\

\noindent {\bf Remarks:}
Now we will discuss how to sample entries of $M$ using sampling method~\eqref{eq:prob} in $O(nnz(M)+ m\log(n))$ time. Consider the following multinomial based sampling model: sample the number of elements per row (say $m_i$) by doing $m$ draws using a multinomial distribution over the rows, given by $\{0.5(\frac{d\|M^i\|^2}{(n+d) \|M\|_F^2}+\frac{1}{n+d})+0.5\frac{\|M^i\|_1}{\|M\|_{1,1}} \}$. Then, sample $m_i$ elements of the row-$i$, using $\{0.5\frac{\|M_j\|^2}{ \|M\|_F^2}+0.5\frac{|M_{ij}|}{\|M\|_{1,1}}\}$ over $ j \in [d]$, with replacement. 

The failure probability in this model is bounded by 2 times the failure probability if the elements are sampled according to~\eqref{eq:prob} \cite{candes2009exact}. Hence, we can instead use the above mentioned multinomial model for sampling. Moreover, $\|M^i\|$, $\|M^i\|_1$ and $\|M_j\|$ can be computed in time $O(nnz(M)+n)$, so $m_i$'s can be sampled efficiently. Moreover, the multinomial distribution for all the rows can be computed in time $O(d+nnz(M))$, $O(d)$ work for setting up the first $\|M_j\|$ term and $nnz(M)$ term for changing the base distribution wherever $M_{ij}$ is non-zero. Hence, the total time complexity is $O(nnz(M)+m\log n)$. 



\subsection{Proof Overview:}

%% file: coherent_noisy.tex
We now present the key steps in our proof of Theorem~\ref{thm:main}. As mentioned in the previous section, our algorithm proceeds in two steps: entry-wise sampling of the given matrix $M$ and then weighted alternating minimization (WAltMin) to obtain a low-rank approximation of $M$. 

Hence, the goal is to analyze the WAltMin procedure, with input samples obtained using \eqref{eq:prob}, to obtain the bounds in Theorem~\ref{thm:main}. Now, WAltMin is an iterative procedure solving an inherently non-convex problem, $\min_{U, V}\sum_{(i,j)\in \Omega}w_{ij} (\e_i^TUV^T\e_j-M_{ij})^2$. Hence, it is prone to local minimas or worse, might not even converge. However, recent results for low-rank matrix completion have shown that alternating minimization (with appropriate initialization) can indeed be analyzed to obtain exact matrix completion guarantees. 

Our proof also follows along similar lines, where we show that the initialization procedure (step 4 of Sub-routine~\ref{algo:2}) provides an accurate enough estimate of $M_r$ and then at each step, we show a geometric decrease in distance to $M_r$. However, our proof differs from the previous works in two key aspects: a) existing proof techniques of alternating minimization assume that each element is sampled uniformly at random, while we can allow biased and approximate sampling, b) existing techniques crucially use the assumption that $M_r$ is incoherent, while our proof avoids this assumption using  the weighted version of AltMin.

We now present our bounds for initialization as well as for each step of the WAltMin procedure. Theorem~\ref{thm:main} follows easily from the two bounds. 
\begin{lemma}[Initialization]\label{lem:approx_init}
Let the set of entries $\Omega$ be generated according to \eqref{eq:prob}. Also, let $m \geq C \frac{n}{\delta^2} \log(n)$. Then, the following holds (w.p. $\geq 1-\frac{2}{n^{10}}$): 
\begin{equation}
\lV \Ro(M) -M \rV \leq \delta \lV M \rV_F.
\end{equation}
Also, if $\|M-M_r\|_F \leq \frac{1}{576\kappa r^{1.5}}\|M_r\|_F$, then the following holds (w.p. $\geq 1-\frac{2}{n^{10}}$): 
\begin{align*}\|(\widehat{U}^{(0)})^i\| \leq 8\sqrt{r} \sqrt{ \|M^i\|^2/\|M\|_F^2} ~\text{ and }~  dist(\widehat{U}^{(0)}, \Uo) \leq \frac{1}{2},\end{align*}
where $\widehat{U}^{(0)}$ is the initial iterate obtained using Steps 4, 5 of Sub-Procedure~\ref{algo:2}. $\kappa=\sigma_1^*/\sigma_r^*$, $\sigma_i^*$ is the $i$-th singular value of $M$, $M_r=\Uo\Sigma^*(\Vo)^T$. 
\end{lemma}

Let  $\mathcal{P}_r(A)$ be the best rank-$r$ approximation of $A$.  Then, Lemma~\ref{lem:approx_init} and Weyl's inequality implies that: \begin{multline}\|M-\mathcal{P}_r(\Ro(M))\|  \leq \|M-\Ro(M)\|  +\|\Ro(M)-\mathcal{P}_r(\Ro(M))\|  \leq \|M-M_r\|+2\delta \|M\|_F.  \label{eq:init_final}\end{multline}

Now, we can have two cases: 1) $\|M-M_r\|_F \geq \frac{1}{576\kappa r^{1.5}}\|M_r\|_F$: In this case, setting $\delta=\epsilon/(\kappa r^{1.5})$ in \eqref{eq:init_final} already implies the required error bounds of Theorem~\ref{thm:main}\footnote{There is a small technicality here: alternating minimization can potentially worsen this bound. But the error after each step of alternating minimization can be effectively checked using a small cross-validation set and we can stop if the error increases.}. 2) $\|M-M_r\|_F \leq \frac{1}{576\kappa r^{1.5}}\|M_r\|_F$. In this regime, we will show now that alternating minimization reduces the error from initial $\delta\|M\|_F$ to $\delta \|M-M_r\|_F$.  

\begin{lemma}[WAltMin Descent]\label{lem:waltmin_descent}
Let hypotheses of Theorem~\ref{thm:main} hold. Also, let $\|M-M_r\|_F \leq \frac{1}{576\kappa r\sqrt{r}}\|M_r\|_F$. Let $\widehat{U}^{(t)}$ be the $t$-th step iterate of Sub-Procedure~\ref{algo:2} (called from Algorithm~\ref{alg:1}), and let $\widehat{V}^{(t+1)}$ be the $(t+1)$-th iterate (for $V$). Also, let $\|(\Ut)^i\| \leq 8\sqrt{r}\kappa \sqrt{ \frac{\|M_j\|^2}{\|M\|_F^2}+ \frac{|M_{ij}|}{ \|M\|_F}}$ and $dist({U}^{(t)}, \Uo) \leq \frac{1}{2}$, where $U^{(t)}$ is a set of orthonormal vectors spanning $\widehat{U}^{(t)}$. Then, the following holds (w.p. $\geq 1-\gamma/T$): 
$$dist({V}^{(t+1)}, V^*)\leq \frac{1}{2}dist({U}^{(t)}, \Uo)+ \epsilon \|M-M_r\|_F/\so_r,$$
and $\|(\Vt)^j\| \leq 8\sqrt{r}\kappa \sqrt{ \frac{\|M_j\|^2}{\|M\|_F^2}+ \frac{|M_{ij}|}{ \|M\|_F}}$, where $V^{(t+1)}$ is a set of orthonormal vectors spanning $\widehat{V}^{(t+1)}$.
\end{lemma}
The above lemma shows that distance between $\widehat{V}^{(t+1)}$ and $\Vo$ (and similarly, $\widehat{U}^{(t+1)}$ and $\Uo$) decreases geometrically up to $\epsilon\|M-M_r\|_F/\so_r$. Hence, after $\log(\|M\|_F\|/\zeta)$ steps, the first error term in the bounds above vanishes and the error bound given in Theorem~\ref{thm:main} is obtained. 

Note that, the sampling distribution used for our result is a ``hybrid'' distribution combining leverage scores and the $L_1$-style sampling. However, if $M$ is indeed a rank-$r$ matrix, then our analysis can be extended to handle the leverage score based sampling itself ($q_{ij}=m\cdot\frac{\|M^i\|^2+\|M_j\|^2}{2n\|M\|_F^2}$). Hence our results also show that weighted alternating minimization can be used to solve the coherent-matrix completion problem introduced in \cite{chen2014coherent}. 

%% file: covariance.tex
\def\qoij{q_{ij}}
\def\qohij{\hat{q}_{ij}}
\def\qtij{q_{ij}}
\def\qthij{\hat{q}_{ij}}




\subsection{Direct Low-rank Approximation of Matrix Product}\label{sec:covariance}

In this section we present a new pass efficient algorithm for the following problem: suppose we are given two matrices, and desire a low-rank approximation of their product $AB$; in particular, we are {\em not} interested in the actual full matrix product itself (as this may be unwieldy to store and use, and thus wasteful to produce in its entirety). One example setting where this arises is when one wants to calculate the joint counts between two very large sets of entities; for example, web companies routinely come across settings where they need to understand (for example) how many users both searched for a particular query and clicked on a particular advertisement. The number of possible queries and ads is huge, and finding this co-occurrence matrix from user logs involves multiplying two matrices -- query-by-user and user-by-ad respectively -- each of which is itself large. 

We give a method that directly produces a low-rank approximation of the final product, and involves storage and manipulation of only the efficient factored form (i.e. one tall and one fat matrix) of the final intended low-rank matrix. Note that as opposed to the previous section, the matrix does not already exist and hence we do not have access to its row and column norms; so we need a new sampling scheme (and a different proof of correctness).

{\bf Algorithm:} Suppose we are given an $n_1 \times d$ matrix $A$ and another $d\times n_2$ matrix $B$, and we wish to calculate a rank-$r$ approximation of the product $A\cdot B$. Our algorithm proceeds in two stages:
\begin{enumerate}
\item Choose a biased random set $\Omega\subset [n_1]\times [n_2]$ of elements as follows: choose an intended number $m$ (according to Theorem \ref{thm:mult} below) of sampled elements, and then independently include each $(i,j) \in  [n_1]\times [n_2]$ in $\Omega$ with probability given by $\hat{q}_{ij}=\min\{1,q_{ij}\}$ where
\begin{equation}
  \label{eq:prob_mp}
  q_{ij} ~ := ~ m\cdot\left(\frac{\|A^i\|^2}{n_2\|A\|_F^2} + \frac{\|B_j\|^2}{n_1 \|B\|_F^2}\right),
\end{equation}
Then, find $P_\Omega(A\cdot B)$, i.e. only the elements of the product $AB$ that are in this set $\Omega$. 
\item Run the alternating minimization procedure WAltMin$(P_{\Omega}(A\cdot B), \Omega, r, \hat{q}, T)$, where $T$ is the number of iterations (again chosen according to Theorem \ref{thm:mult} below). This produces the low-rank approximation in factored form.
\end{enumerate}

{\bf Remarks:} Note that the sampling distribution now depends only on the row norms $\|A^i\|^2$ of $A$ and the column norms $\|B_j\|^2$ of $B$; each of these can be found completely in parallel, with one pass over each row/column of the matrices $A$ / $B$. A second pass, again parallelizable, calculates the element $(A\cdot B)_{ij}$ of the product, for $(i,j)\in \Omega$. Once this is done, we are again in the setting of doing weighted alternating minimization over a small set of samples -- the setting we had before, and as already mentioned this too is highly parallelizable and very fast overall. In particular, the computation complexity of the algorithm is $O(|\Omega|\cdot (d+r^2))=O(m(d+r^2))=O(\frac{nr^3 \kappa^2}{\eps^2} \cdot (d+r^2) )$ (suppressing terms dependent on norms of $A$ and $B$ ), where $n= \max \{n_1, n_2 \}$.

We now present our theorem on the number of samples and iterations needed to make this procedure work with at least a constant probability. 

\begin{theorem}
\label{thm:mult}
Consider matrices $A\in \R^{n_1\times d}$ and $B\in \R^{d\times n_2}$ and let $m \, =\, \frac{C}{\gamma}\cdot \frac{(\|A\|_F^2+\|B\|_F^2)^2}{\|AB\|_F^2}\cdot \frac{n r^3}{(\eps)^2} \kappa^2 \log(n) \log^2(\frac{\|A\|_F+\|B\|_F}{\zeta})$, where $\kappa=\sigma_1^*/\sigma_r^*$, $\sigma_i^*$ is the $i$-th singular value of $A\cdot B$ and $T=\log(\frac{\|A\|_F+\|B\|_F}{\zeta})$. Let $\Omega$ be sampled using probability distribution \eqref{eq:prob_mp}. Then, the output $\widehat{AB}_r=WAltMin(P_{\Omega}(A\cdot B), \Omega, r, \hat{q}, T)$ of Sub-routine~\ref{algo:2} satisfies (w.p. $\geq 1-\gamma$):  $\qquad \|A\cdot B-\widehat{AB}_r\| \leq  \|A\cdot B- (A\cdot B)_r\| +\eps\|A\cdot B- (A\cdot B)_r\|_F+\zeta.$ 
\end{theorem}

Next, we show an application of our matrix-multiplication approach to approximating covariance matrices $M=YY^T$, where $Y$ is a $n \times d$ sample matrix. Note, that a rank-$r$ approximation to $M$ can be computed by computing low rank approximation of $Y$, $\widehat{Y}_r$, i.e., $\tilde{M}_r=\widehat{Y}_r \widehat{Y}_r^T$. However, as we show below, such an approach leads to weaker bounds as compared to computing low rank approximation of $YY^T$: 

\begin{corollary}\label{cor:cov}
Let $M=YY^T\in \R^{n\times n}$ and let $\Omega$ be sampled using probability distribution \eqref{eq:prob_mp} with $m \geq \frac{C}{\gamma}\frac{n r^3}{\eps^2} \kappa^2 \log(n) \log^2(\frac{\|Y\|}{\zeta})$, the output $\widehat{M}_r$ of $WAltMin(P_{\Omega}(M), \Omega, r, \hat{q}, \log(\frac{\|Y\|}{\zeta}))$ satisfy (w.p. $\geq 1-\gamma$): \[\|\widehat{M}_r-(YY^T)_r\| \leq \eps  \lV Y-Y_r \rV_F^2+\zeta.\]
Further when $\|Y-Y_r\|_F \leq \|Y_r\|_F$ we get:  \[\|\widehat{M}_r- (YY^T)_r\| \leq \eps  \lV YY^T-(YY^T)_r \rV_F+\zeta.\]
\end{corollary}
Now, one can compute $\tilde{M}_r=\widehat{Y}_r \widehat{Y}_r^T$ in time $O(n^2r + \frac{nr^5}{\eps^2})$ with error $\frac{\|(YY^T)_r-\tilde{M}_r\|}{\|Y\|^2} \leq \eps\frac{\|Y-Y_r\|_F }{\|Y\|}$. Where as computing low rank approximation of $YY^T$ gives(from Corollary~\ref{cor:cov})  $\frac{\|(YY^T)_r-\widehat{M}_r\|}{\|Y\|^2} \leq \eps\frac{\|YY^T-Y_r Y_r^T\|_F}{\|Y\|^2}$, which can be much smaller than $\eps\frac{\|Y-Y_r\|_F}{\|Y\|}$.  The difference in error is a consequence of larger gap between singular values of $YY^T$ compared to $Y$. For related discussion and applications see section 4.5 of ~\cite{halko2011finding}.

%% file: distributedPCA.tex
\def\yk{Y_{\mathbf{c_k}}}
\def\yhk{\widehat{Y}_{\mathbf{c_k}}}
\def\vhk{\widehat{V}_{\mathbf{c_k}}}
\def\Uhk{\widehat{U}_{\mathbf{r_k}}}
\section{Distributed Principal Component Analysis}

\begin{algorithm}[t]
\caption{Distributed low rank approximation algorithm}
\label{alg:dpca}
\begin{algorithmic}[1]
\INPUT Matrix $\Mk$ at server $k$, rank-$r$, number of samples $m$ and number of iterations $T$.
\STATE {\it Sampling:} Each server $k$ computes column norms of $\Mk$,  $\|\Mk\|_{1,1}$ and communicates to $CP$. CP computes column norms of $M$, $\|M\|_{1,1}$, $\|M\|_F$ and communicates to all servers.
\STATE Each server $k$ samples $(i, j)$th entry with probability $\min\{ m(\frac{\|\Mk^i\|^2 +\|M_j\|^2}{2n\|M\|_F^2} +\frac{(\Mk)_{ij}}{\|M\|_{1,1}}) ,1\}$ for rows $\{ r_k \}$ and generates $\Omega^k$.
\STATE Each server $k$ sends lists of columns ($\{ c_k \} \subset [d]$) where $\Rok(\Mk)$ has sampled entries, to CP.
\STATE {\it Initialization:} CP generates random $n \times r$ matrix $Y^{(0)}$ and communicates $\yk^{(0)}$ to server $k$.
\FOR {$t=0$ to $\log(1/c)$}
	\STATE Each server $k$ computes $\yhk^{(t+1)} = \Rok(\Mk)^T \Rok(\Mk) *  \yk^{(t)}$ and communicates to CP.
	\STATE CP computes $Y^{(t+1)} =\sum_{k} \yhk^{(t+1)}$, normalize and communicates $\yk^{(t+1)}$ to server $k$.
\ENDFOR
\STATE {\it WAltMin:} Each server $k$ set  $\widehat{V}_{\mathbf{c_k}}^{(0)} =\yk^{(t+1)}$.
\FOR {$t=0$ to $T-1$}
	\STATE Each server $k$ computes $(\widehat{U}^{(t+1)})^i =\arg \min_{x \in \R^r} \sum_{j: (i, j) \in \Omega^k} \wij \left( M_{ij} -x^T(\widehat{V}^{t})^j \right)^2$,  for all $i \in \{ r_k \}$.
	\STATE Each server $k$ sends to CP; $z_{kj} =\left(\Uhk^{(t+1)} \right)^T \Rok(\Mk)e_j$ and $B_{kj}=\sum_{i: (i,j) \in \Omega^k} \wij u u^T$, $u= (\widehat{U}^{(t+1)})^i $ for all $j \in \{ c_k \}$.
	\STATE CP computes $B_j= \sum_k B_{kj}$ and $(\Vht)^j = B_j^{-1} \sum_k z_{kj}$ for $j=1, \cdots, d$ and communicates $(\Vht)_{\mathbf{c_k}}$ to server $k$.
\ENDFOR
\OUTPUT Server $k$ has $\Uhk^{(t+1)}$ and CP has $\Vht$. 
\end{algorithmic}
\end{algorithm}

Modern large-scale systems have to routinely compute PCA of data matrices with millions of data points embedded in similarly large number of dimensions. Now, even storing such matrices on a single machine is not possible and hence most industrial scale systems use distributed computing environment to handle such problems. However, performance of such systems depend not only on computation and storage complexity, but also on the required amount of communication between different servers. 

In particular, we consider the following distributed PCA setting: Let $M\in \R^{n\times d}$ be a given matrix (assume $n\geq d$ but $n\approx d$). Also, let $M$ be row partitioned among $s$ servers and let $\Mk\in \R^{n \times d}$ be the matrix with rows $\{r_k\}\subseteq [n]$ of $M$ and rest filled with zeros, stored on $k$-th server. Moreover, we assume that one of the servers act as Central Processor(CP) and in each round all servers communicate with the CP and the CP communicates back with all the servers. Now, the goal is to compute $\widehat{M}_r$, an estimate of $M_r$, such that the total communication (i.e. number of bits transferred) between CP and other servers is minimized. Note that, such a model is now standard for this problem and was most recently studied by \cite{kannan2014principal}.




Recently several interesting results~\cite{feldman2013turning, ghashami2014relative, liang2013distributed, kannan2014principal} have given algorithms to compute rank-$r$ approximation of $M$, $\tilde{M}_r$ in the above mentioned distributed setting.  In particular, \cite{kannan2014principal} proposed a method that for row-partitioned model  requires $O(\frac{dsr}{\eps} + \frac{ sr^2}{\eps^4})$ communication to obtain a relative Frobenius norm guarantee, $$||M-\tilde{M}_r||_F \leq (1 +\eps)||M-M_r||_F.$$

In contrast, a distributed setting extension of our LELA algorithm~\ref{alg:1} has linear communication complexity $O(ds + \frac{n r^5}{\eps^2})$ and computes rank-$r$ approximation $\widehat{M}_r$, with $||M-\widehat{M}_r|| \leq ||M-M_r|| +\eps||M-M_r||_F$. Now note that if $n\approx d$ and if $s$ scales with $n$ (which is a typical requirement), then our communication complexity can be significantly better than that of \cite{kannan2014principal}. Moreover, our method provides spectral norm bounds as compared to relatively weak Frobenius bounds mentioned above. \\

\noindent {\bf Algorithm:} The distributed version of our LELA algorithm depends crucially on the following observation: given $V$, each row of $U$ can be updated independently. Hence, servers need to communicate rows of $V$ only. There also, we can use the fact that each server requires only $O(nr/s\log n)$ rows of $V$ to update their corresponding $\Uk$. $\Uk$  denote restriction of ${U}$ to rows in set $\{ r_k \}$ and $0$ outside and similarly $\vhk, \yhk$ denote restriction of $V, \widehat{Y}$ to rows $\{c_k\}$.  

 We now describe the distributed version of each of the critical step of LELA algorithm. See Algorithm~\ref{alg:dpca} for a detailed pseudo-code. For simplicity, we dropped the use of different set of samples in each iteration. Correspondingly the algorithm will modify to distributing samples $\Omega^k$ into $2T+1$ buckets and using one in each iteration. This simplification doesn\rq{}t change the communication complexity.


{\em Sampling}: For sampling, we first compute column norms $\|M_j\|, \forall 1\leq j\leq d$ and communicate to each server. This operation would require $O(ds)$ communication. Next, each server (server $k$) samples elements from its rows $\{\mathbf{r_k}\}$ and stores $\Rok(\Mk)$ locally. Note that because of {\em independence over rows}, the servers don't need to transmit their samples to other servers. 

{\em Initialization}: In the initialization step, our algorithm computes top $r$ right singular vector of $\Ro(M)$ by iterations $\widehat{Y}^{(t+1)} =\Ro(M)^T\Ro(M) Y^t=\sum_k \Rok(M)^T\Rok(M)Y^t$. Now, note that computing $\Rok(M) Y^t$ requires server $k$ to access atmost $|\Omega^k|$ columns of $Y^t$. Hence, the total communication from the CP to all the servers in this round is $O(|\Omega|r)$. Similarly, each column of $\Rok(M)^T\Rok(M)Y^t$ is only $|\Omega^k|$ sparse. Hence, total communication from all the servers to CP in this round is $O(|\Omega|r)$. Now, we need constant many rounds to get a constant factor approximation to SVD of $\Ro(M)$, which is enough for good initialization in WAltMin procedure. Hence, total communication complexity of the initialization step would be $O(|\Omega|r)$. 

{\em Alternating Minimization Step}: For alternating minimization,  update to rows of $U$ is computed at the corresponding servers and the update to $V$ is computed at the CP. For updating $\widehat{U}^{(t+1)}_{\mathbf{r_k}}$ at server $k$, we use the following observation: updating $\widehat{U}^{(t+1)}_{\mathbf{r_k}}$ requires atmost $|\Omega^k|$ rows of $\widehat{V}^{(t)}$. Hence, the total communication from CP to all the servers in the $t$-th iteration is $O(|\Omega|r)$. Next, we make a critical observation that  update $\widehat{V}^{(t+1)}$  can be computed by adding certain messages from each server (see Algorithm~\ref{alg:dpca}  for more details). Message from server $k$ to CP is of size $O(|\Omega^k|r^2)$. Hence, total communication complexity in each round is $O(|\Omega|r^2)$ and total number of rounds is $O(\log(\|M\|_F/\zeta))$.

We now combine the above given observations to provide error bounds and communication complexity of our distributed PCA algorithm:
\begin{theorem}\label{thm:dpca}
Let the $n \times d$ matrix $M$ be distributed over $s$ servers according to the row-partition model. Let $m \geq \frac{C}{\gamma} \frac{n r^3}{\eps^2} \kappa^2 \log(n) \log^2(\frac{||M_r||}{\zeta}) $. Then, the algorithm~\ref{alg:dpca} on completion will leave matrices $\Uhk^{(t+1)}$ at server $k$ and $\Vht$ at CP such that the following holds (w.p. $\geq 1-\gamma$):  $||M- \widehat{U}^{(t+1)} (\Vht)^T || \leq  ||M-M_r|| + \eps ||M-M_r||_F +\zeta$, where $\widehat{U}^{(t+1)}= \sum_k \Uhk^{(t+1)}$. This algorithm has a communication complexity of $O(ds+|\Omega|r^2)=O(ds+ \frac{n r^5\kappa^2}{\eps^2} \log^2(\frac{||M_r||}{\zeta}))$ real numbers.
\end{theorem}
As discussed above,  each update to $\widehat{V}^{(t)}$ and  $\widehat{U}^{(t)}$ are computed exactly as given in the WAltMin procedure (Sub-routine~\ref{algo:2}). Hence, error bounds for the algorithm follows directly from Theorem~\ref{thm:main}. Communication complexity bounds follows by observing that  $|\Omega|\leq 2m$ w.h.p.\\

\noindent {\bf Remark:} The sampling step given above suggests another simple algorithm where we can compute $P_{\Omega}(M)$ in a distributed fashion and communicate the samples to CP. All the computation is performed at CP afterwards. Hence, the total communication complexity would be $O(ds+|\Omega|)=O(ds+\frac{n r^3\kappa^2}{\eps^2} \log(\frac{||M_r||}{\zeta})$, which is lesser than the communication complexity of Algorithm~\ref{alg:dpca}. However, such an algorithm is not desirable in practice, because it is completely reliant on one single server to perform all the computation. Hence it is slower and is  fault-prone. In contrast, our algorithm can  be implemented in a peer-peer scenario as well and is more fault-tolerant.




Also, the communication complexity bound of Theorem~\ref{thm:dpca} only bounds the total number of real numbers transferred. However, if each of the number requires several bits to communicate then the real communication can still be very large. Below, we bound each of the real number that we transfer, hence providing a bound on the number of bits transferred.\\ 

\noindent {\bf Bit complexity:} First we will bound $\wij M_{ij}$. Note that we need to bound this only for $(i, j) \in \Omega$. Now, $|\wij M_{ij}|\leq ||\Ro(M)||_{\infty} \leq ||\Ro(M)|| \leq ||M|| +\eps||M||_F \leq 2*ndM_{max}$, where the third inequality follows from Lemma~\ref{lem:approx_init}. Hence if the entries of the matrix $M$ are being represented using $b$ bits initially then the algorithm needs to use $O(b + \log(nd))$ bits.  By the same argument we get a bound of $O(b + \log(nd))$ bits for computing $||M^i||^2, \forall i; ||M_j||^2, \forall j; ||M||_F^2$ and $||M||_{1,1}$.

Further at any stage of the WAltMin iterations $||\Uht (\Vht)^T||_{\infty} \leq ||\Uht (\Vht)^T|| \leq 2||M||_F$. So this stage also needs $O(b + \log(n))$ bits for computation. Hence overall the bit complexity of each of the real numbers of the algorithm~\ref{alg:dpca} is $O(b + \log(nd))$ , if $b$ bits are needed for representing the matrix entries. That is, overall communication complexity of the algorithm is $O((b+\log (nd))\cdot (ds+\frac{n r^3\kappa^2}{\eps^2} \log(\frac{||M_r||}{\zeta}))$.



%% file: simulations.tex
\section{Simulations}\label{sec:simulations}

\begin{figure*}[ht]
  \centering
  \begin{tabular}[ht]{ccc}
    \includegraphics[width=.5\textwidth]{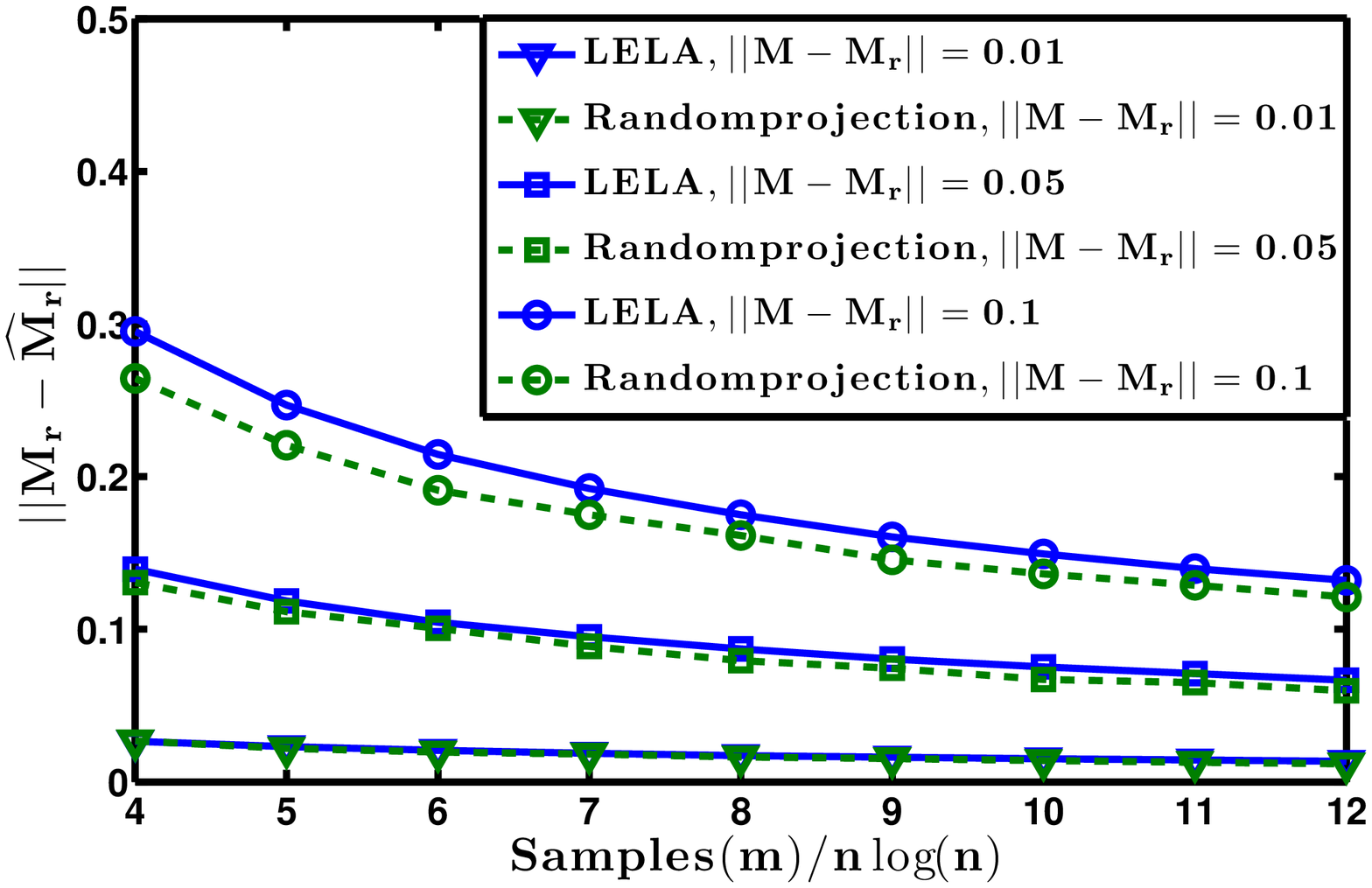}\hspace*{-0pt}&\includegraphics[width=.5\textwidth]{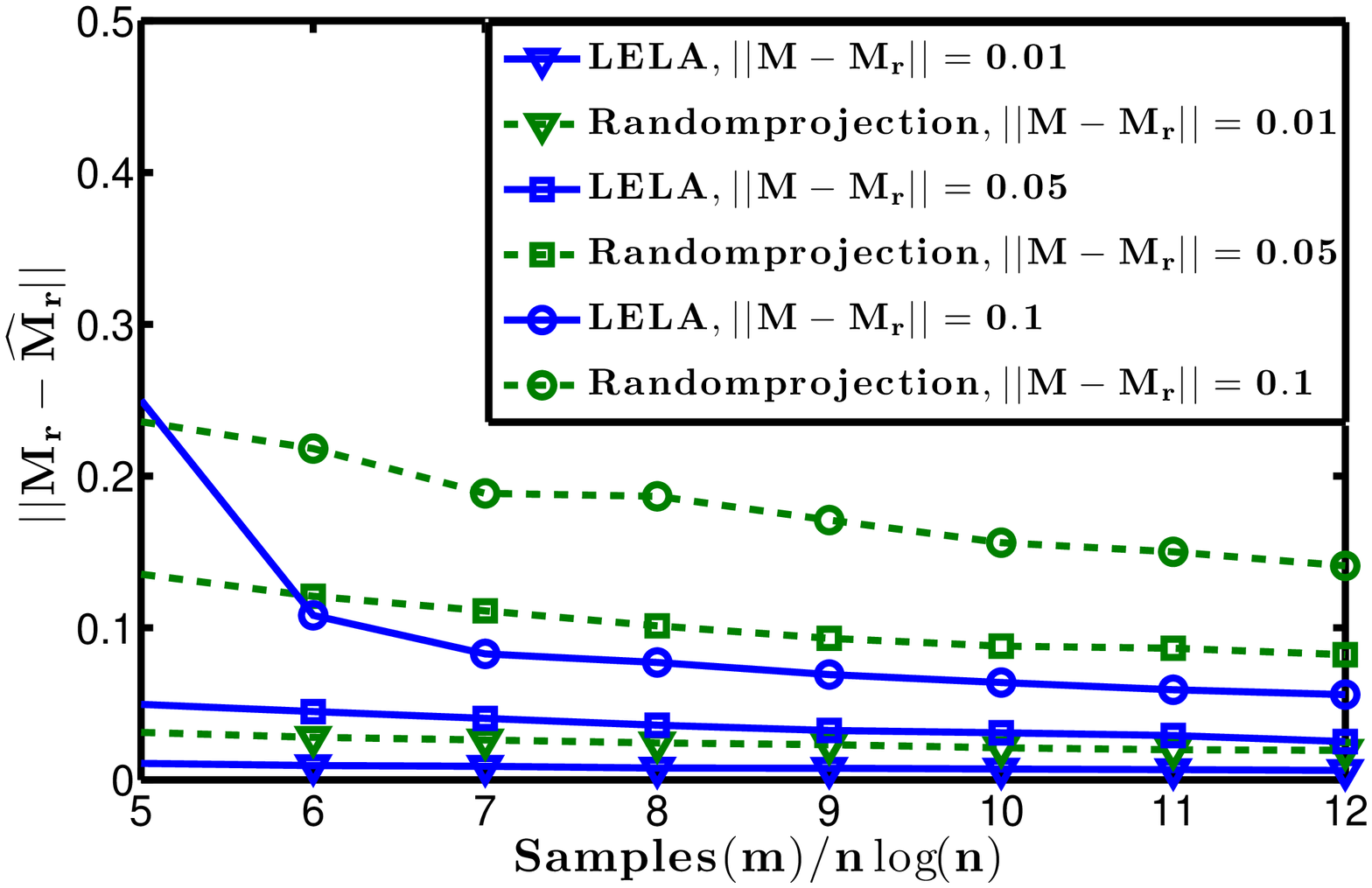}\hspace*{-15pt}\\
{\bf (a)}&{\bf (b)}
  \end{tabular}
  \caption{ Figure plots how the error $||M_r -\widehat{M}_r||$ decreases with increasing number of samples $m$ for different values of noise $||M-M_r||$, for incoherent and coherent matrices respectively. Algorithm LELA~\ref{alg:1} is run with $m$ samples and Gaussian projection algorithms is run with corresponding dimension of the projection $l=m/n$. Computationally LELA algorithms takes $O(nnz(M) +m\log(n) )$ time for computing samples and Gaussian projection algorithm takes $O(nm)$ time. {\bf(a)}:For same number of samples both algorithms have almost the same error for incoherent matrices. {\bf (b):} For coherent matrices clearly the error of LELA algorithm (solid lines) is much smaller than that of random projection (dotted lines).}
  \label{fig:plot12}
\end{figure*} 

\begin{figure*}[ht]
  \centering
  \begin{tabular}[ht]{ccc}
    \includegraphics[width=.5\textwidth]{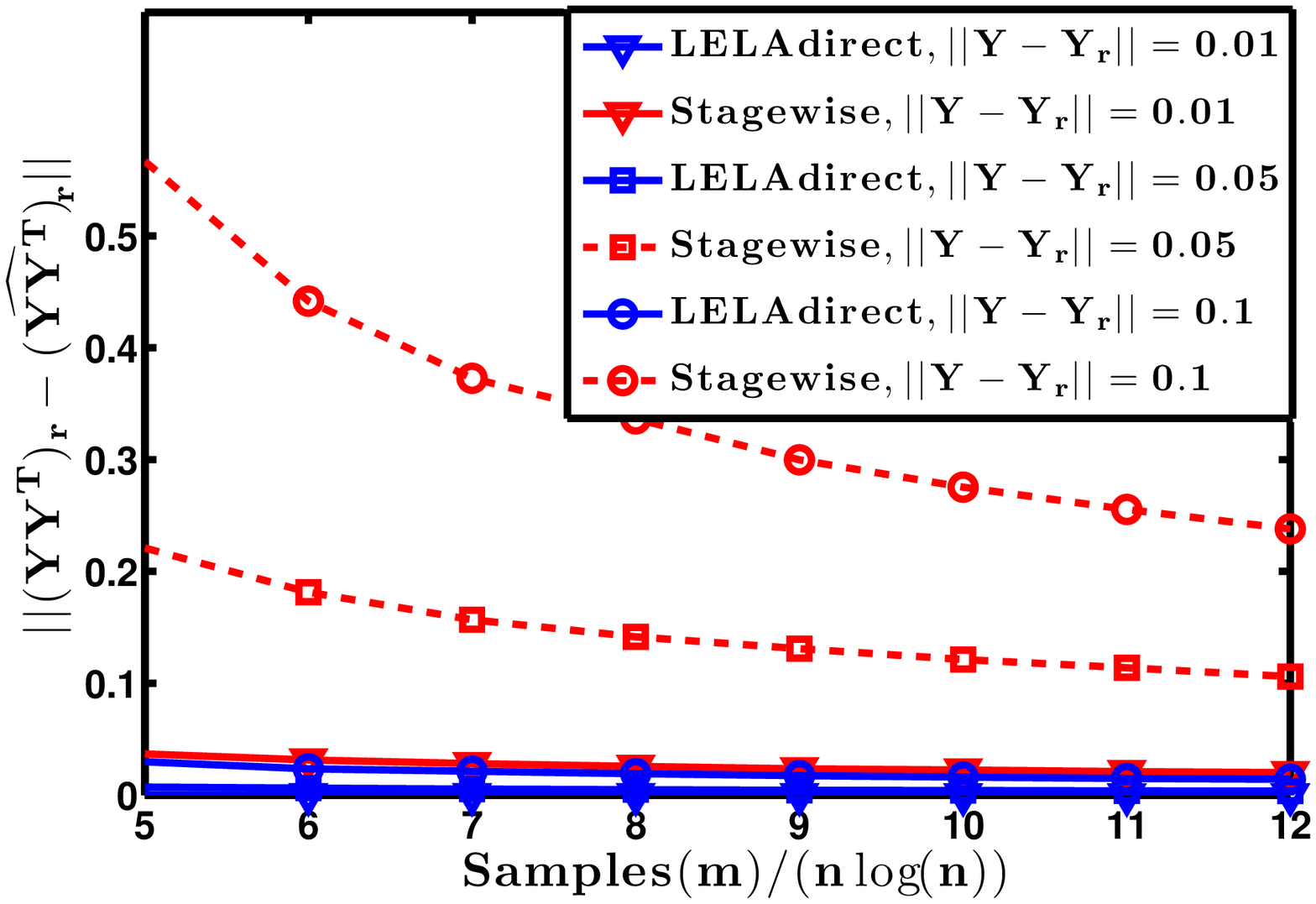}\hspace*{-0pt}&\includegraphics[width=.5\textwidth]{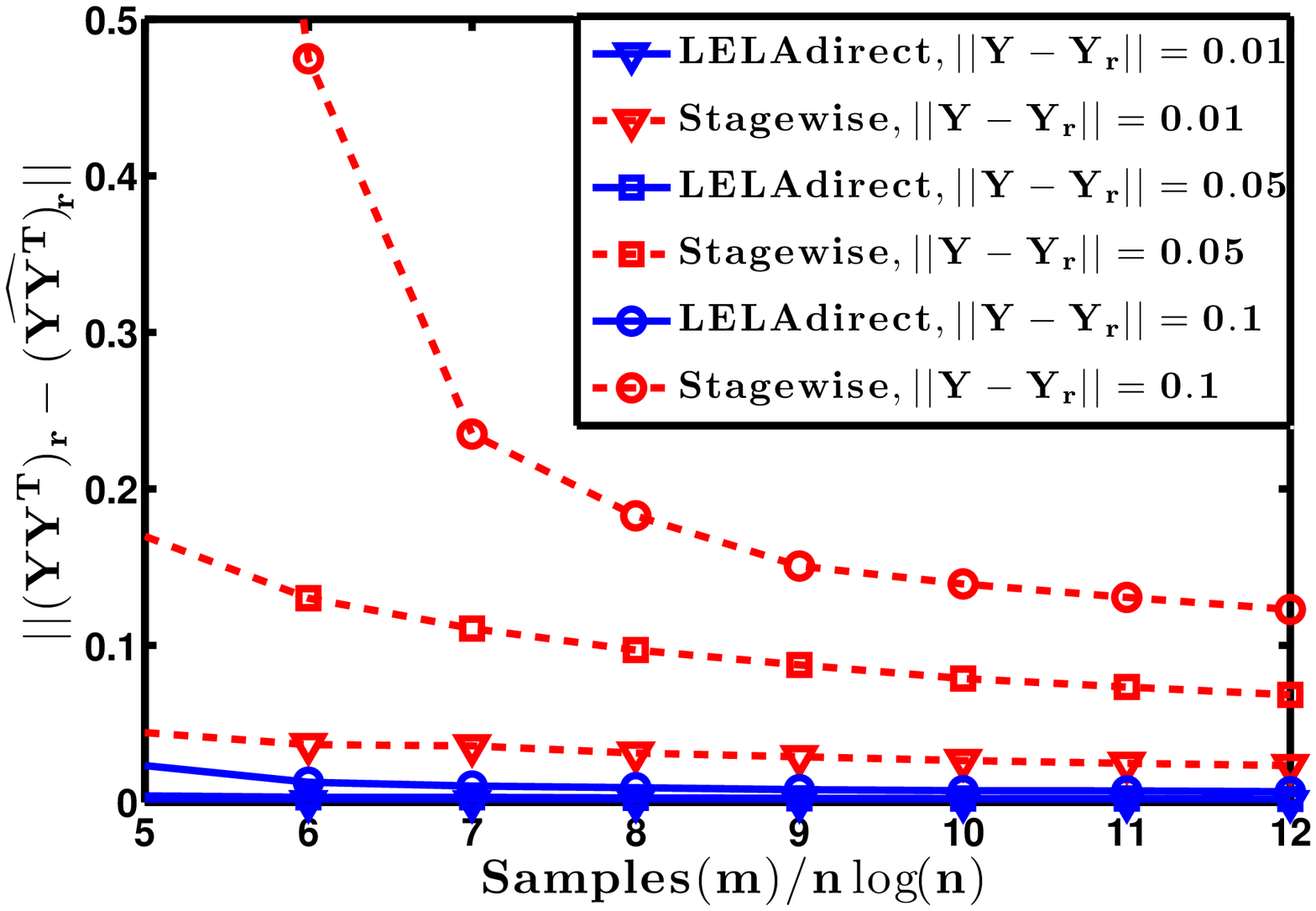}\hspace*{-15pt}\\
{\bf (a)}&{\bf (b)}
  \end{tabular}
  \caption{ Figure plots the error $||(YY^T)_r- \widehat{(YY^T)}_r||$ for LELA direct (Section~\ref{sec:covariance}) and Stagewise algorithm for {\bf(a)}:incoherent matrices and for {\bf (b):} coherent matrices. Stagewise algorithm is first computing rank-$r$ approximation $\widehat{Y}_r$ of $Y$ using algorithm~\ref{alg:1} and setting the low rank approximation $\widehat{(YY^T)}_r =\widehat{Y}_r \widehat{Y}_r^T.$ Clearly directly computing low rank approximation of $YY^T$ has smaller error.}
  \label{fig:plot34}
\end{figure*} 

\begin{figure*}[ht]
  \centering
  \begin{tabular}[ht]{ccc}
    \includegraphics[width=.5\textwidth]{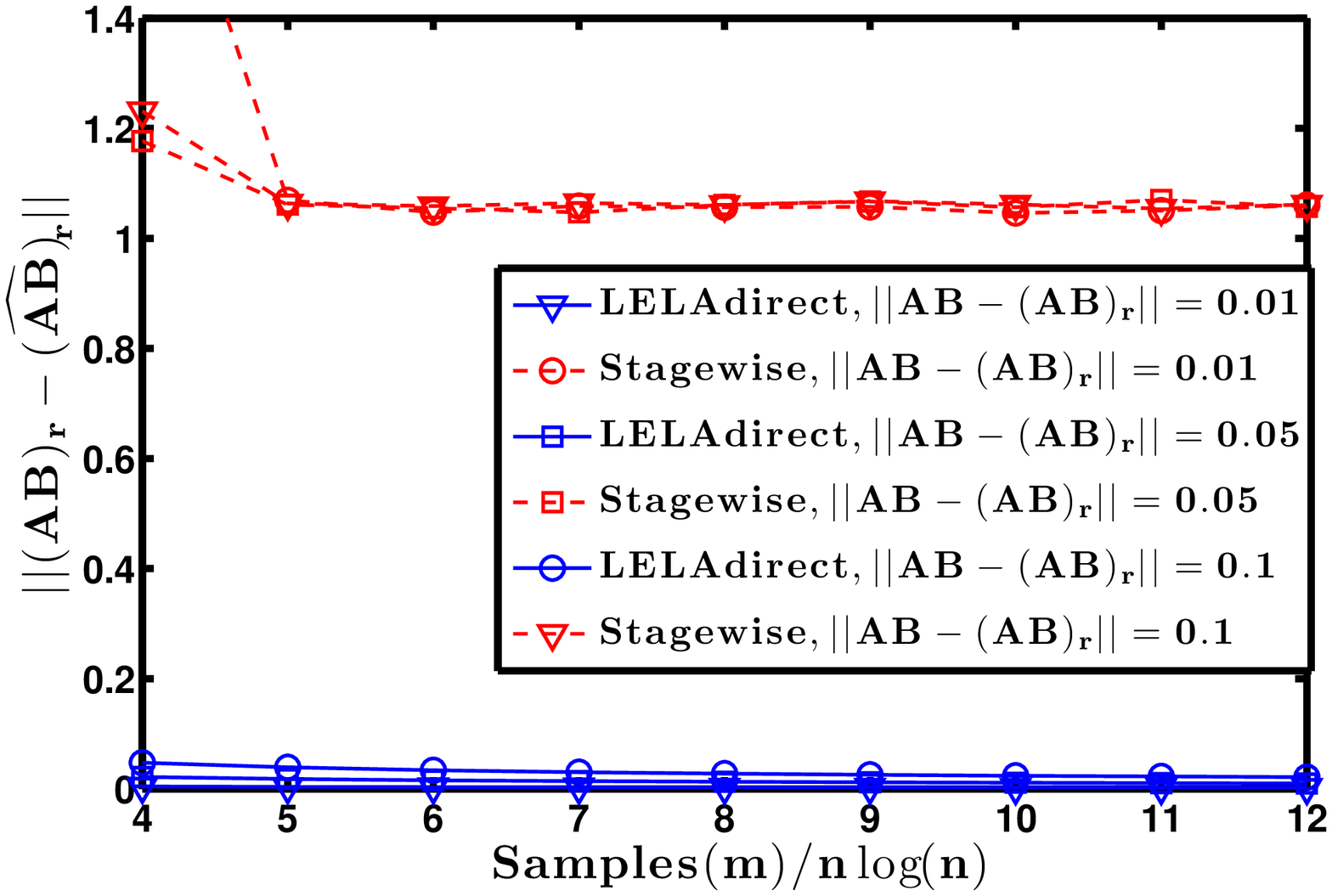}\hspace*{-0pt}&\includegraphics[width=.5\textwidth]{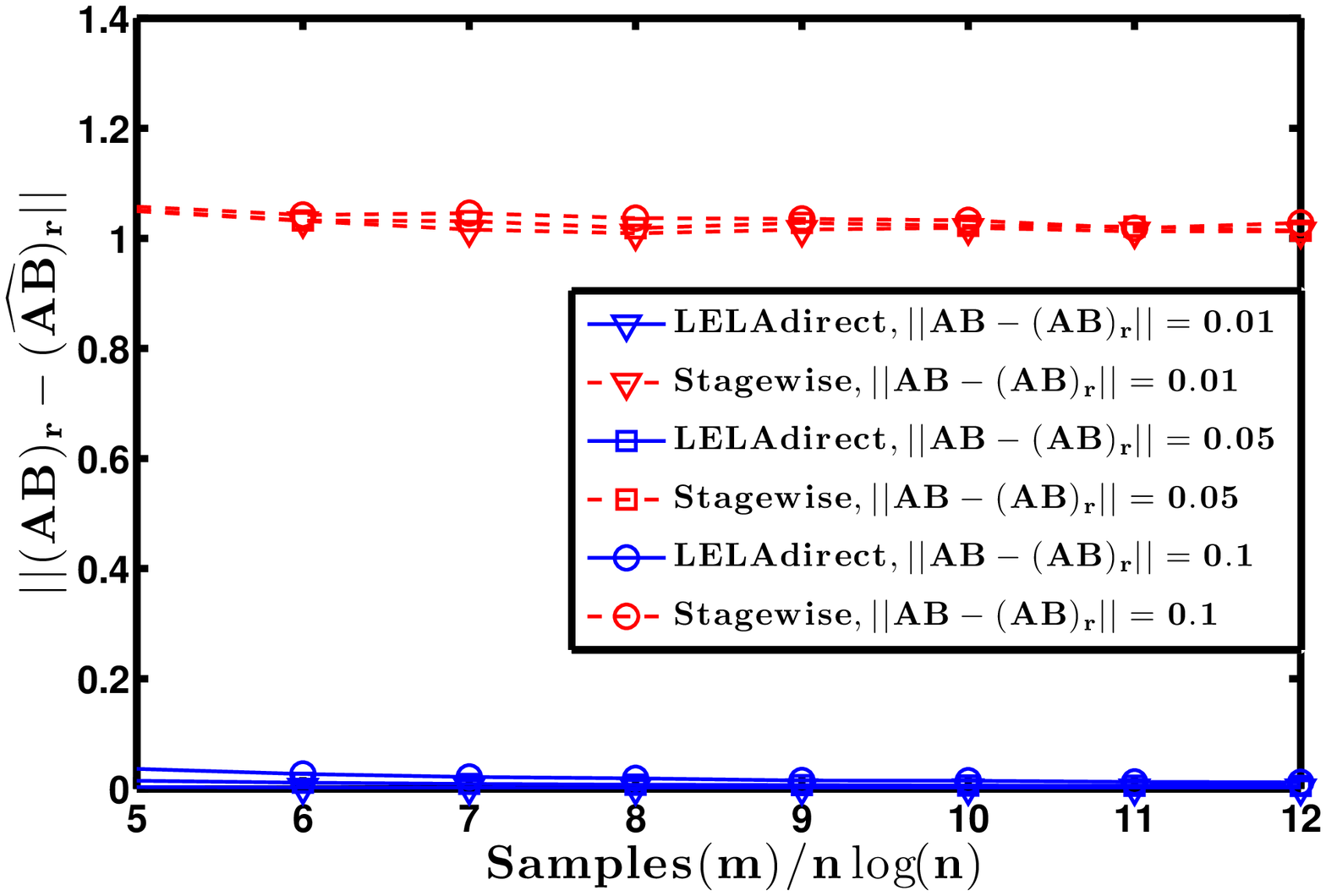}\hspace*{-15pt}\\
{\bf (a)}&{\bf (b)}
  \end{tabular}
  \caption{Figure plots the error $||(AB)_r- \widehat{(AB)}_r||$ for LELA direct (Section ~\ref{sec:covariance}) and Stagewise algorithm for {\bf(a)}:incoherent matrices and for {\bf (b):} coherent matrices. Stagewise algorithm is first computing rank-$r$ approximation $\widehat{A}_r, \widehat{B}_r$ of $A, B$ respectively using algorithm~\ref{alg:1} and setting the low rank approximation $\widehat{(AB)}_r =\widehat{A}_r \widehat{B}_r.$ Clearly directly computing low rank approximation of $AB$ has smaller error.}
  \label{fig:plot56}
\end{figure*}

In this section we present some simulation results on synthetic data to show the error performance of the algorithm~\ref{alg:1}. First we consider the setting of finding low rank approximation of a given matrix $M$. Later we consider the setting of computing low rank approximation of $A \cdot B$, given $A$ and $B$ without computing the product.

For simulations we consider random matrices of size 1000 by 1000 and rank-5. $M_r$ is a rank 5 matrix with all singular values 1. We consider two cases, one in which $M_r$ is incoherent and other in which $M_r$ is coherent. Recall that a $n \times d$ rank-$r$ matrix $M_r$ is an incoherent matrix if $\|(\Uo)^i\|^2 \leq \frac{\mu_0 r}{n}, \forall i$ and $\|(\Vo)^j\|^2 \leq \frac{\mu_0 r}{d}, \forall j$, where SVD of $M_r$ is $\Uo \So (\Vo)^T$. Intuitively incoherent matrices have mass spread over almost all entries whereas coherent matrices have mass concentrated on only few entries.

To generate matrices with varying incoherence parameter $\mu_0$, we use the power law matrices model ~\cite{chen2014coherent}. $M_r =DUV^T D$, where $U$ and $V$ are random $n \times r$ orthonormal matrices and $D$ is a diagonal matrix with $D_{ii} \propto \frac{1}{i^{\alpha}}$. For $\alpha=0$  $M_r$ is an incoherent matrix with $\mu_0=O(1)$ and for $\alpha=1$ $M_r$ is a coherent matrix with $\mu_0 =O(n)$.

The input to algorithms is the matrix $M=M_r + Z$, where $Z$ is a Gaussian noise matrix with $||Z||=0.01, 0.05$ and $0.1$. Correspondingly Frobenius norm of $Z$ is $||Z||*\sqrt{1000}/2$, which is $0.16, 0.79$ and $1.6$ respectively.

Each simulation is averaged over 20 different runs. We run the WAltMin step of the algorithm for 15 iterations. Note that using different set of samples in each iteration of WAltMin subroutine~\ref{algo:2} is generally observed to be not required in practice. Hence we use the same set of samples for all iterations.  

In the first plot we compare the error $||M_r -\widehat{M}_r||$ of our algorithm LELA~\ref{alg:1} with the random projection based algorithm~\cite{halko2011finding, boutsidis2013improved}. We use the matrix with each entry an independent Gaussian random variable as the sketching matrix, for the random projection algorithm. Other choices are Walsh-Hadamard based transform~\cite{woolfe2008fast} and sparse embedding matrices~\cite{clarkson2013low}. 

We compare the error of both algorithms as we vary number of samples $m$ for algorithm~\ref{alg:1}, equivalently varying the dimension of random projection $l=m/n$ for the random projection algorithm. 
In figure \ref{fig:plot12} we plot the error $||M_r -\widehat{M}_r||$ with varying number of samples $m$ for both the algorithms. For incoherent matrices we see that LELA algorithm has almost the same error as the random projection algorithm~\ref{fig:plot12}(a). But for coherent matrices we notice that in figure \ref{fig:plot12}(b) LELA has significantly smaller error. 

Now we consider the setting of computing low rank approximation of $YY^T$ given $Y$ using algorithm LELA direct discussed in section \ref{sec:covariance} with sampling \eqref{eq:prob_mp}. In figure~\ref{fig:plot34} we compare this algorithm with a stagewise algorithm, which computes low rank approximation $\widehat{Y}_r$ from $Y$ first and sets the rank-$r$ approximation of $YY^T$ as $\widehat{Y}_r \widehat{Y}^T_r$. As discussed in section~\ref{sec:covariance} direct approximation of $YY^T$ has less error than that of computing $\widehat{Y}_r \widehat{Y}^T_r$. Again plot~\ref{fig:plot34}(a) is for incoherent matrices and plot~\ref{fig:plot34}(b) is for coherent matrices.

Finally in figure~\ref{fig:plot56} we consider the case where $A$ and $B$ are two rank $2r$ matrices with $AB$ being a rank $r$ matrix. Here the top $r$ dimensional row space of $A$ is orthogonal to the top $r$ dimensional column space of $B$. Hence simple algorithms that compute rank $r$ approximation of $A$ and $B$ first and then multiply will have high error as compared to that of LELA direct.

%% file: appendix.tex
\def\cab{C_{AB}}
\section{Concentration Inequalities}
In this section we will review couple of concentration inequalities we use in the proofs.

\begin{lemma}[Bernstein's Inequality]\label{lem:bernstein}
Let $X_1,...X_n$ be independent scalar random variables. Let $|X_i| \leq L, \forall i ~w.p. ~1$. Then,
\begin{equation}
\prob{\lv\sum_{i=1}^n X_i -\sum_{i=1}^n \expec{X_i}\rv \geq t} \leq 2\exp\left(\frac{-t^2/2}{\sum_{i=1}^n \Var(X_i)+ Lt/3}\right). 
\end{equation}
\end{lemma}

\begin{lemma}[Matrix Bernstein's Inequality~\cite{tropp2012user}]\label{lem:mbernstein}
Let $X_1,...X_p$ be independent random matrices in $\rnn$. Assume each matrix has bounded deviation from its mean: \[ \|X_i - \expec{X_i}\| \leq L, \forall i ~w.p. ~1.\] Also let the variance be \[ \sigma^2 =\max \left\{ \lV \expec{\sum_{i=1}^p (X_i-\expec{X_i}) (X_i- \expec{X_i})^T \rV}, \lV \expec{\sum_{i=1}^p (X_i - \expec{X_i})^T (X_i - \expec{X_i})}\rV \right\}.\] Then,
\begin{equation}
\prob{\lV\sum_{i=1}^n \left(X_i - \expec{X_i}\right)\rV \geq t} \leq 2n\exp\left(\frac{-t^2/2}{\sigma^2 + Lt/3}\right). 
\end{equation}
\end{lemma}

\noindent Recall the Shatten-$p$ norm of a matrix $X$ is \begin{equation*}\|X\|_p = \left(\sum_{i=1}^n \sigma_i(X)^p \right)^{1/p}. \end{equation*} $ \sigma_i(X)$ is the $i$th singular value of $X$. In particular for $p=2$, Shatten-$2$ norm is the Frobenius norm of the matrix.
\begin{lemma}\label{lem:mchebyshev}[Matrix Chebyshev Inequality~\cite{mackey2012matrix}]
Let $X$ be a random matrix. For all $t >0$,
\begin{equation}
\prob{\lV X \rV \geq t} \leq \inf_{p \geq 1} t^{-p} \expec{\lV X \rV_p^p}.
\end{equation}
\end{lemma}

\section{Proofs of section~\ref{sec:lela}}\label{sec:pfnoisy}
In this section we will present proof for Theorem~\ref{thm:main}. For simplicity we will only discuss proofs for the case when matrix is square. Rectangular case is a simple extension. We will provide proofs of the supporting lemmas first.

We will recall some of the notation now. $$q_{ij}=m\cdot \left(0.5\frac{\|M^i\|^2+\|M_j\|^2}{2n \|M\|_F^2}+0.5\frac{|M_{ij}|}{\|M\|_{1,1}}\right).$$ Let $\qhij =\min\{\qij, 1\}.$ This is to make sure the probabilities are all less than 1. Recall the definition of weights $\wij =1/\qhij$ when $\qhij >0$ and $0$ else. Note that $\sum_{ij}\qhij \leq m$. Also let $m \geq \beta nr\log(n)$. 

Let $\{\dij\}$ be the indicator random variables and $\dij=1$ with probability $\qhij$. Define $\Omega$ to be the sampling operator with $\Omega_{ij} =\dij$. Define the weighted sampling operator $\Ro$ such that, $\Ro(M)_{ij} = \dij w_{ij} M_{ij}$. 

Throughout the proofs we will drop the subscript of $\Omega$ that denotes different sampling sets at each iteration of WAltMin.

First we will abstract out the properties of the sampling distribution~\eqref{eq:prob} that we use in the rest of the proof. 

\begin{lemma}\label{lem:approx_supp1}
For $\Omega$ generated according to~\eqref{eq:prob} and under the assumptions of Lemma~\ref{lem:approx_init} the following holds, for all $(i ,j)$ such that $ \qij \leq 1$.
\begin{equation}\label{eq:prop_a1} \frac{M_{ij}}{\qhij} \leq \frac{2n}{m}\|M\|_F,\end{equation} 
\begin{equation}\label{eq:prop_a2} \sum_{ \{j:  \qhij =\qij \}} \frac{M_{ij}^2}{\qhij} \leq \frac{4n}{m} \|M\|_F^2, \end{equation}
\begin{equation}\label{eq:prop_a3} \frac{\|(\Uo)^i\|^2}{\qhij} \leq \frac{8nr\kappa^2}{m}, \end{equation} and
\begin{equation}\label{eq:prop_a4} \frac{\|(\Uo)^i\| \|(\Vo)^j\|}{\qhij}  \leq \frac{8nr\kappa^2}{m}. \end{equation}  
\end{lemma}

The proof of the lemma~\ref{lem:approx_supp1} is straightforward from the definition of $\qij$.

 \subsection{Initialization}

Now we will provide proof of the initialization lemma~\ref{lem:approx_init}.\\
{\bf Proof of lemma~\ref{lem:approx_init}:}
\begin{proof}
The proof of this lemma has two parts. \\
1) We show that $$\lV \Ro(M) -M \rV \leq \delta \lV M \rV_F$$\\
2) We show that the trimming step of algorithm~\ref{algo:2} gives the required row norm bounds on $\widehat{U}^{(0)}$. $$\|(\widehat{U}^{(0)})^i\| \leq 8\sqrt{r} \sqrt{ \|M^i\|^2/\|M\|_F^2} ~\text{ and }~ dist(\widehat{U}^{(0)}, \Uo) \leq \frac{1}{2},$$

{\it Proof of the first step:}
We prove the  proof of the first part using the matrix Bernstein inequality. Note that the $L1$ term in the sampling distribution will help in getting good bounds on absolute magnitude of random variables $X_{ij}$ in this proof. 

Let $X_{ij} = (\dij -\qhij) \wij M_{ij} e_i e_j^T$. Note that $\{ X_{ij} \}_{i,j=1}^n$ are independent zero mean random matrices. Also $\Ro(M) -\expec{\Ro(M)} = \sum_{ij}X_{ij}$.

First we will bound $\|X_{ij}\|$. When $\qij \geq 1$, $\qhij=1$ and $\dij=1$, and $X_{ij} =0$ with probability 1. Hence we only need to consider cases when $\qhij=\qij \leq 1$. We will assume this in all the proofs without explicitly mentioning it any more.
\begin{align} \|X_{ij}\| = \max \{ \lv(1- \qhij) \wij M_{ij} \rv, \lv \qhij \wij M_{ij}\rv \}. \end{align}  Recall $\wij =1/\qhij$. Hence
\begin{align*}
  \lv(1- \qhij) \wij M_{ij}\rv  &=  \lv(\frac{1}{\qhij}- 1)  M_{ij}\rv  \leq \lv \frac{M_{ij}}{\qhij} \rv  \stackrel{\zeta_1}{\leq} \frac{2n}{m}\|M\|_F.
\end{align*}
$\zeta_1$ follows from~\eqref{eq:prop_a1}.
\begin{align*}
 \lv \qhij \wij M_{ij}\rv  = \lv M_{ij} \rv \stackrel{\zeta_1}{\leq} \lv\frac{M_{ij}}{\qhij} \rv \leq \frac{2n}{m}\|M\|_F.
\end{align*}
$\zeta_1$ follows from $\qhij \leq 1$.

Hence, $\|X_{ij}\| $ is bounded by $L = \frac{2n}{m}\|M\|_F$. Now we will bound the variance.

\begin{align*}
\lV \expec{ \sum_{ij} X_{ij} X_{ij}^T} \rV &= \lV \expec{ \sum_{ij} (\dij -\qhij)^2 \wij^2 M_{ij}^2 e_i e_i^T} \rV =\lV  \sum_{ij}  \qhij(1- \qhij) \wij^2 M_{ij}^2 e_i e_i^T \rV \\
&= \max_i \lv  \sum_{j} \qhij(1- \qhij) \wij^2 M_{ij}^2  \rv.
\end{align*}
Now,
\begin{align*}
 \sum_{j} \qhij(1- \qhij) \wij^2 M_{ij}^2 =\sum_j (\frac{1}{\qhij}-1) M_{ij}^2 \leq \sum_j \frac{ M_{ij}^2}{(\qhij)}  \stackrel{\zeta_1}{\leq} \frac{4n}{m}\|M\|_F^2.
\end{align*}
$\zeta_1$ follows from~\eqref{eq:prop_a2}.
Hence \begin{align*} \lV \expec{ \sum_{ij} X_{ij} X_{ij}^T} \rV = \max_i \lv  \sum_{j} \qhij(1- \qhij) \wij^2 M_{ij}^2  \rv \leq \max_i  \frac{4n}{m}\|M\|_F^2 = \frac{4n}{m}\|M\|_F^2. \end{align*} We can prove the same bound for the $\lV \expec{ \sum_{ij} X_{ij}^T X_{ij}} \rV$. Hence $\sigma^2 = \frac{4n}{m}\|M\|_F^2$. Now using matrix Bernstein inequality with $t=\delta \|M\|_F$ gives, with probability $ \geq 1-\frac{2}{n^{2}}$,
\begin{equation*}
\lV \Ro(M) -\expec{\Ro(M)}\rV =\lV \Ro(M) -M\rV \leq  \delta \lV M \rV_F.
\end{equation*}

Hence we get $\|M-\mathcal{P}_r(\Ro(M))\| \leq \|M-\Ro(M)\|+\|\Ro(M)-\mathcal{P}_r(\Ro(M))\|\leq \|M-M_r\|+2\delta \|M\|_F$, which implies $\|M_r -\mathcal{P}_r(\Ro(M))\| \leq 2\|M-M_r\|+2\delta \|M\|_F$.
 
Let SVD of $\mathcal{P}_r(\Ro(M))$ be $U^{(0)} \Sigma^{(0)} (V^{(0)})^T $. Hence,
 \begin{align*}
 \|\mathcal{P}_r(\Ro(M)) -M_r\|^2 &=\| U^{(0)} \Sigma^{(0)} (V^{(0)})^T -  \Uo \So (\Vo)^T\|^2 \\
 &=\| U^{(0)} \Sigma^{(0)} (V^{(0)})^T - U^{(0)}(U^{(0)})^T \Uo \So (\Vo)^T-(I - U^{(0)}(U^{(0)})^T) \Uo \So (\Vo)^T\|^2 \\
 &\geq \|(I - U^{(0)}(U^{(0)})^T) \Uo \So (\Vo)^T\|^2\\
 &\geq (\so_r)^2 \| (U^{(0)}_{\perp})^T \Uo\|^2.
 \end{align*}
 This implies  $dist(U^{(0)}, \Uo) \leq \frac{2\|M-M_r\|+2\delta \|M\|_F}{\so_r} \leq  \frac{1}{144r}.$ This follows from the assumption $\|M-M_r\|_F \leq \frac{1}{576 \kappa r^{1.5}} \|M_r\|_F$ and $\delta \leq \frac{1}{576\kappa r^{1.5}}$. $\kappa =\frac{\so_1}{\so_r}$ is the the condition number of $M_r$. 

{\it Proof of the trimming step:}
From previous step we know that $\|\Ro(M)- M\| \leq  \delta \lV M \rV_F$ and consequently $dist(U^{(0)}, \Uo) \leq \delta_2$. Let, \begin{align*} l_i = \sqrt{4 \|M^i\|^2/\|M\|_F^2} , \end{align*} be the estimates for the left leverages scores of the matrix $M$. Since $\|M-M_r\|_F \leq  \|M_r\|_F$, $ l_i^2 \geq \frac{\sum_{k=1}^r (\so_k)^2 (\Uo_{ik})^2}{\sum_{k=1}^r (\so_k)^2}$.
 
Set the elements bigger than $2l_i$ in the $ith$ row of $U^0$ to 0 and let $\tilde{U}$ be the new initialization matrix obtained. Also since $dist(U^{(0)} , \Uo) \leq \delta_2$,  for every $j=1,..,r$ there exists a vector $\bar{u}_j$ in $\rn$, such that $\ip{U^{(0)}_j}{\bar{u}_j} \geq \sqrt{1-\delta_2^2}$, $\|\bar{u}_j\| =1$ and $|(\bar{u}_j)_i|^2 \leq \frac{\sum_{k=1}^r (\so_k)^2 (\Uo_{ik})^2}{\sum_{k=1}^r (\so_k)^2}$ for all $i$. Now $\tilde{U}_j$ is the $jth$ column of $\tilde{U}$ obtained by setting the entries of the $jth$ column of $U^{(0)}$ to $0$ whenever the $ith$ entry of $U^{(0)}_j$ is bigger than $2l_i$. For such $i$,
\begin{align} 
\lv (\tilde{U}_j)_i - (\bar{u}_j)_i \rv \leq  \sqrt{\frac{\sum_{k=1}^r (\so_k)^2 (\Uo_{ik})^2}{\sum_{k=1}^r (\so_k)^2}} \leq  \lv (U^0_j)_i -  (\bar{u}_j)_i \rv,
\end{align}
since $\lv (U^{(0)}_j)_i -  (\bar{u}_j)_i \rv \geq 2l_i -l_i  = \sqrt{\frac{\sum_{k=1}^r (\so_k)^2 (\Uo_{ik})^2}{\sum_{k=1}^r (\so_k)^2}} $.

For the rest of the coordinates, $(\tilde{U}_j)_i = (U^{(0)}_j)_i$. Hence, 
\begin{align*} 
\lV \tilde{U}_j -\bar{u}_j \rV \leq  \lV U^{(0)}_j - \bar{u}_j \rV = \sqrt{1 +1 -2\ip{ U^{(0)}_j}{\bar{u}_j }} \leq \sqrt{2} \delta_2.
\end{align*}
Hence $\lV \tilde{U}_j \rV \geq 1-\sqrt{2}\delta_2$, and $\lV U^{(0)}_j - \tilde{U}_j \rV \leq \sqrt{1 -\lV \tilde{U}_j \rV^2} \leq 2\sqrt{\delta_2},$ for $\delta_2 \leq \frac{1}{\sqrt{2}}$. Also $\lV U^{(0)} - \tilde{U} \rV_F \leq 2\sqrt{r \delta_2}$. This gives abound on the smallest singular value of $\tilde{U}$.
\[ \sigma_{\min}(\tilde{U}) \geq \sigma_{\min}(U^{(0)}) - \sigma_{\max}(\tilde{U} -U^{(0)}) \geq 1- 2\sqrt{r \delta_2}.\]

Now let the reduced QR decomposition of $\tilde{U}$ be $\tilde{U} = \widehat{U}^{(0)} \Lambda^{-1}$, where $\widehat{U}^{(0)}$ is the matrix with orthonormal columns. From the bounds above we get that 
\begin{align*}
\|\Lambda\|^2 = \frac{1}{\sigma_{\min}(\Lambda^{-1})^2} =\frac{1}{\sigma_{\min}(\widehat{U}^{(0)} \Lambda^{-1})^2} =\frac{1}{\sigma_{\min}(\tilde{U})^2} \leq 4,
\end{align*}
when $\delta_2 \leq \frac{1}{16r}$.

First we will show that this trimming step will not increase the distance to $\Uo$ by much. To bound the $dist(\widehat{U}^{(0)}, \Uo)$ consider, $\|(\uo_{\perp})^T U\|$, where $\uo_{\perp}$ is some vector perpendicular to $\Uo$.
\begin{align*}
\|(\uo_{\perp})^T \widehat{U}^{(0)}\| =\|(\uo_{\perp})^T \tilde{U} \Lambda\| &\leq (\|(\uo_{\perp})^T U^{(0)}\| +\|(\uo_{\perp})^T (\tilde{U}-U^{(0)})\|)\|\Lambda\| \\
&\leq (\delta_2 + 2\sqrt{r \delta_2})2 \leq 6\sqrt{r\delta_2} \leq \frac{1}{2},
\end{align*}
for $\delta_2 \leq \frac{1}{144r}$. Second we will bound $\|(\widehat{U}^{(0)})^i\|$.
\begin{align*}
\|(\widehat{U}^{(0)})^i\| =\|e_i^T \widehat{U}^{(0)}\|=\|e_i^T \tilde{U} \Lambda\| \leq \|e_i^T \tilde{U}\| \| \Lambda\|  \leq 2 l_i \sqrt{r} 2 \leq 8\sqrt{r} \sqrt{ \|M^i\|^2/\|M\|_F^2} .
\end{align*}
Hence we finish the proof of the second part of the lemma.
\end{proof}

\subsection{Weighted AltMin analysis}\label{sec:proofs_waltmin_descent}
We first provide proof of Lemma~\ref{lem:waltmin_descent} for rank-1 case to explain the main ideas and in the next section we will discuss rank-$r$ case. Hence $M_1 =\so \uo (\vo)^T$. Before the proof of the lemma we will prove couple of supporting lemmas.

Let $\ut$ and $\vt$ be the normalized vectors of the iterates $\uht$ and $\vht$ of WAltMin.  We assume that samples for each iteration are generated independently. For simplicity we will drop the subscripts on $\Omega$ that denote different set of samples in each iteration in the rest of the proof. 

The weighted alternating minimization updates at the $t+1$ iteration are, \begin{equation} \label{eq:witerates}\|\uht\| \widehat{v}^{t+1}_j = \so\vo_j \frac{\sum_i \dij \wij \ut_i \uo_i }{\sum_i \dij \wij (\ut_i)^2} +\frac{\sum_i \dij \wij \ut_i (M-M_1)_{ij}}{\sum_i \dij \wij (\ut_i)^2}.\end{equation} Writing in terms of power method updates we get,  \begin{equation}\label{eq:witerates1} \|\uht\| \widehat{v}^{t+1} =\so\ip{\uo}{\ut}\vo- \so B^{-1} (\ip{\ut}{\uo}B -C)\vo  +B^{-1} y,\end{equation} where $B$ and  $C$  are diagonal matrices with $B_{jj} = \sum_i  \dij \wij (\ut_i)^2$ and  $C_{jj} =\sum_i \dij \wij \ut_i \uo_i$ and $y$ is the vector $\Ro(M-M_1)^T \ut$ with entries $y_{j}=\sum_i \dij \wij \ut_i (M-M_1)_{ij}$. 

 Now we will bound the error caused by the $M-M_1$ component in each iteration. 
 \begin{lemma}\label{lem:noise_samplebound}
 For $\Omega$ generated according to~\eqref{eq:prob} and under the assumptions of Lemma~\ref{lem:waltmin_descent} the following holds:
 \begin{equation}
 \lV (\ut)^T\Ro(M-M_1) - (\ut)^T (M-M_1) \rV \leq \delta \|M-M_1\|_F,
 \end{equation}
  with probability greater that $1-\frac{\gamma}{T\log(n)}$, for $m \geq \beta n \log(n)$, $\beta \geq \frac{4 c_1^2 T}{\gamma \delta^2}$. Hence,$\lV (\ut)^T\Ro(M-M_1)  \rV \leq dist(\ut, \uo)\|M-M_1\| + \delta \lV M-M_1 \rV_F,$ for constant $\delta$.
 \end{lemma}

\begin{proof}[Proof of lemma~\ref{lem:noise_samplebound}]
Let the random matrices $X_{ij} = (\dij-\qhij)\wij (M-M_1)_{ij} (\ut)_i e_j^T$. Then $\sum_{ij} X_{ij} =(\ut)^T\Ro(M-M_1) - (\ut)^T (M-M_1)$. Also $\expec{X_{ij} } =0$. We will use the matrix Chebyshev inequality for $p=2$. Now we will bound $\expec{\lV\sum_{ij}X_{ij}\rV_2^2}$. \begin{align*} \expec{\lV\sum_{ij}X_{ij}\rV_2^2} &= \expec{\sum_j \left(\sum_i  (\dij-\qhij)\wij (M-M_1)_{ij} (\ut)_i \right)^2} \\
&\stackrel{\zeta_1}{=} \sum_j \sum_i  \qhij(1-\qhij)(\wij)^2 (M-M_1)_{ij}^2 (\ut_i)^2  \\
&\leq  \sum_{ij} \wij (\ut_i)^2 (M-M_1)_{ij}^2 \\
&\stackrel{\zeta_2}{\leq} \frac{4n c_1^2}{m} \|M-M_1\|_F^2.
\end{align*}
$\zeta_1$ follows from the fact that $X_{ij}$ are zero mean independent random variables. $\zeta_2$ follows from~\eqref{eq:infty_tapprox}. Hence applying the matrix Chebyshev inequality for $p=2$ and $t=\delta \|M-M_1\|_F$ gives the result.
\end{proof}

\begin{lemma}\label{lem:approx_supp2}
For $\Omega$ sampled according to~\eqref{eq:prob} and under the assumptions of Lemma~\ref{lem:waltmin_descent} the following holds: 
\begin{equation}
\lv \sum_j \dij \wij (\uo_j)^2- \sum_j (\uo_j)^2 \rv \leq \delta_1,
\end{equation}
 with probability greater that $1-\frac{2}{n^{2}}$, for $m \geq \beta n \log(n)$, $\beta \geq \frac{16}{\delta_1^2}$ and $\delta_1 \leq 3$.
\end{lemma}

\begin{proof}[Proof of Lemma~\ref{lem:approx_supp2}]
Note that $\lv \sum_j \dij \wij (\uo_j)^2- \sum_j (\uo_j)^2 \rv=\lv \sum_j (\dij -\qhij) \wij (\uo_j)^2 \rv$. Let the random variable $X_{j} = (\dij -\qhij) \wij (\uo_j)^2 $. $\expec{X_j} =0$ and $\Var(X_j) = \qhij (1- \qhij) (\wij (\uo_j)^2)^2.$ Hence,
\begin{align*}
\sum_j \Var(X_j) &= \sum_j \qhij (1- \qhij) (\wij (\uo_j)^2)^2 = \sum_j (\frac{1}{\qhij} -1) (\uo_j)^4 \leq \sum_j \frac{(\uo_j)^2}{\qhij} (\uo_j)^2\\
&\stackrel{\zeta_1}{\leq} \frac{16n}{m}\sum_j  (\uo_j)^2 =\frac{16n}{m}.
\end{align*}
$\zeta_1$ follows from~\eqref{eq:prop_a3}.
Also it is easy to check that $|X_j| \leq \frac{16n}{m}$. Now applying Bernstein inequality gives the result.
\end{proof}

\begin{lemma}\label{lem:approx_supp3}
For $\Omega$ sampled according to~\eqref{eq:prob} and under the assumptions of Lemma~\ref{lem:waltmin_descent} the following holds: 
\begin{equation}
\|(\ip{\ut}{\uo}B -C)\vo\| \leq \delta_1 \sqrt{1 -\ip{\uo}{\ut}^2},
\end{equation}
with probability greater than $1-\frac{2}{n^2}$,  for $m \geq \beta n\log(n), \beta \geq \frac{48 c_1^2 }{\delta_1^2}$  and $\delta_1 \leq 3$.
\end{lemma}

\begin{proof}
Let \[\alpha_{i} =\ut_i(\ip{\uo}{\ut}\ut_i-   \uo_i).\] Hence the $jth$ coordinate of the error term in equation~\eqref{eq:witerates1} is \[\frac{\sum_i \dij \wij \alpha_{i}\vo_j}{\sum_i \dij \wij (\ut_i)^2}.\] Recall that $\alpha_{i} =\ut_i(\ip{\uo}{\ut}\ut_i-   \uo_i).$ Let $X_{ij} = \dij \wij \alpha_i \vo_j e_j e_1^T$, for $i,  j$ in $[1,..n]$. Note that $X_{ij}$ are independent random matrices. Then $(\ip{\ut}{\uo}B -C)\vo$ is the first and the only column of the matrix $\sum_{i, j=1}^n X_{ij}$. We will bound $\|\sum_{ij=1}^n X_{ij}\|$ using matrix Bernstein inequality. 

\begin{align*}\sum_{ij}\expec{X_{ij}} =\sum_j \sum_{i} \qhij \wij \alpha_i \vo_j e_j e_1^T =\sum_j \sum_{i} \alpha_i \vo_j e_j e_1^T=0, \end{align*} because $\sum_i \alpha_i =0$. Now we will give a bound on $\|X_{ij}\|$. \begin{align*} \|X_{ij}\| &=|\dij \wij \alpha_i \vo_j|  \leq  \frac{ \vo_j\ut_i}{\qhij} (\ip{\uo}{\ut}\ut_i-   \uo_i) \stackrel{\zeta_1}{\leq}  \frac{16nc_1 }{m} \sqrt{\sum_i (\ip{\uo}{\ut}\ut_i-   \uo_i)^2}\\
& = \frac{16nc_1}{m}  \sqrt{1-\ip{\uo}{\ut}^2}\end{align*} $\zeta_1$ follows from~\eqref{eq:prop_a4} and~\eqref{eq:infty_tapprox}.

Now we will bound the variance. \begin{align*}\lV\expec{\sum_{ij} (X_{ij}-\expec{X_{ij}})^T (X_{ij}-\expec{X_{ij}})}\rV &=\lV\expec{\sum_{ij} (X_{ij}T X_{ij}-\expec{X_{ij}}^T \expec{X_{ij}})}\rV \\ &= \lV\sum_{j} \sum_i \qhij(1-\qhij) (\wij \alpha_i \vo_j)^2 e_1e_1^T\rV.\end{align*}  Then, \begin{align*} \sum_i  \qhij(1-\qhij) (\wij \alpha_i \vo_j)^2 &\leq \sum_i    \wij (\ut_i)^2 (\ip{\uo}{\ut}\ut_i-   \uo_i) ^2 (\vo_j)^2 \leq   \frac{c_1^2 4n}{m} (\vo_j)^2 \sum_i (\ip{\uo}{\ut}\ut_i-   \uo_i) ^2 \\ &{\leq}\frac{c_1^2 4n}{m} (\vo_j)^2(1-\ip{\uo}{\ut}^2). \end{align*} Hence,
\begin{align*} 
\lV\expec{\sum_{ij} (X_{ij}-\expec{X_{ij}})^T (X_{ij}-\expec{X_{ij}})}\rV \leq \frac{4n c_1^2}{m}(1-\ip{\uo}{\ut}^2).
\end{align*}
The lemma follows from applying matrix Bernstein inequality.
\end{proof}

Now we will provide proof of lemma~\ref{lem:waltmin_descent}.\\
{\bf Proof of lemma~\ref{lem:waltmin_descent}:}[Rank-1 case]
\begin{proof}
Now we will prove that the distance between $\ut$, $\uo$ and $\vt, \vo$ decreases with each iteration. Recall that from the assumptions of the lemma we have the following row norm bounds for $\ut$; \begin{equation}\label{eq:infty_tapprox}|\ut_i| \leq c_1\sqrt{ \|M^i\|^2/\|M\|_F^2 + |M_{ij}|/\|M\|_F} .\end{equation}  First we will prove that $dist(\ut, \uo)$ decreases in each iteration and second we will prove that $v^{t+1}$ satisfies similar bound on its row norms. 

\noindent {\bf Bounding $\ip{\vto}{\vo}$:}

Using Lemma~\ref{lem:approx_supp2}, Lemma~\ref{lem:approx_supp3} and equation~\eqref{eq:witerates1} we get,\begin{align}  \|\uht\|\ip{\widehat{v}^{t+1}}{\vo}  \geq \so\ip{\ut}{\uo} - \so \frac{\delta_1}{1-\delta_1} \sqrt{1 -\ip{\uo}{\ut}^2} - \frac{1}{1-\delta_1}\| y^T \vo\|\end{align} and \begin{align}  \|\uht\|\ip{\widehat{v}^{t+1}}{\vo_{\perp}}  \leq \so \frac{\delta_1}{1-\delta_1} \sqrt{1 -\ip{\uo}{\ut}^2} + \frac{1}{1-\delta_1}\|y \|.\end{align} Hence by applying the noise bounds Lemma~\ref{lem:noise_samplebound} we get,
\begin{align*}
dist(\vto, \vo)^2 &=1-\ip{\vto}{\vo}^2 = \frac{\ip{\widehat{v}^{t+1}}{\vo_{\perp}}^2}{\ip{\widehat{v}^{t+1}}{\vo_{\perp}}^2 + \ip{\widehat{v}^{t+1}}{\vo}^2} \leq  \frac{\ip{\widehat{v}^{t+1}}{\vo_{\perp}}^2}{ \ip{\widehat{v}^{t+1}}{\vo}^2} \\
&\stackrel{\zeta_1}{\leq} \frac{ 4(\delta_1 dist(\ut, \uo) +dist(\ut, \uo) \|M-M_1\|/\so + \delta \|M-M_1\|_F/\so)^2}{ (\ip{\ut}{\uo} -  2\delta_1 \sqrt{1 -\ip{\uo}{\ut}^2} - 2\delta\|M-M_1\|/\so)^2}\\
&\stackrel{\zeta_2}{\leq}\frac{ 4(\delta_1 dist(\ut, \uo) + dist(\ut, \uo)\|M-M_1\|/\so + \delta \|M-M_1\|_F/\so)^2}{ (\ip{\uo}{u^0} -  2\delta_1 \sqrt{1 -\ip{\uo}{u^0}^2}- 2\delta\|M-M_1\|/\so)^2} \\
&\stackrel{\zeta_3}{\leq} 25(\delta_1 dist(\ut, \uo) + dist(\ut, \uo)\|M-M_1\|/\so + \delta \|M-M_1\|_F/\so)^2.
\end{align*}
$\zeta_1$ follows from $\delta_1 \leq \frac{1}{2}$. $\zeta_2$ follows from using $\ip{\ut}{\uo} \geq \ip{u^0}{\uo}$. $\zeta_3$ follows from $(\ip{\uo}{u^0} -  2\delta_1 \sqrt{1 -\ip{\uo}{u^0}^2} \geq \frac{1}{2}$, $\delta \leq \frac{1}{20}$  and $\delta_1 \leq \frac{1}{20}$. Hence \begin{align*} dist(\vto, \vo) &\leq \frac{1}{4}dist(\ut, \uo) +5dist(\ut, \uo)\|M-M_1\|/\so + 5\delta \|M-M_1\|_F/\so  \nonumber\\ &\leq \frac{1}{2}dist(\ut, \uo)  + 5\delta \|M-M_1\|_F/\so.\end{align*}

\vskip 0.2in
\noindent {\bf Bounding $|\vto_j|$:}

From Lemma~\ref{lem:approx_supp2} and~\eqref{eq:infty_tapprox} we get that $\lv \sum_i  \dij \wij (\ut_i)^2 - 1\rv \leq \delta_1$ and $\lv \sum_i  \dij \wij \uo_i \ut_i - \ip{\uo}{\ut}\rv \leq  \delta_1$, when $\beta \geq \frac{16 c_1^2}{\delta_1^2}.$ Hence, \begin{equation}\label{eq:supp_b}1-\delta_1 \leq B_{jj} = \sum_i  \dij \wij (\ut_i)^2 \leq 1 +\delta_1, \end{equation} and  \begin{equation}\label{eq:supp_c}C_{jj} =\sum_i \dij \wij \ut_i \uo_i \leq \ip{\ut}{\uo} +\delta_1. \end{equation}
Recall that  \begin{equation*}  \|\uht\| \lv\widehat{v}_j^{t+1}\rv = \lv\frac{\sum_i \dij \wij \ut_i M_{ij} }{\sum_i \dij \wij (\ut_i)^2}\rv \leq  \frac{1}{1-\delta_1}\sum_i \dij \wij \ut_i M_{ij}.\end{equation*} 

We will bound using $\sum_i \dij \wij \ut_i M_{ij}$ using Bernstein inequality. 

Let $X_i =(\dij-\qhij) \wij \ut_i M_{ij}$. Then $\sum_i \expec{X_i} =0 $ and $\sum_i \ut_i M_{ij} \leq \|M_j\|$ by Cauchy-Schwartz inequality. 

$\sum_i \Var(X_i) = \sum_i \qhij (1-\qhij) (\wij)^2 (\ut_i)^2 M_{ij}^2 \leq \sum_i \wij (\ut_i)^2 M_{ij}^2 \leq \frac{4nc_1^2}{m} \|M_j\|^2$. Finally $|X_{ij}| \leq \lv \wij \ut_i M_{ij} \rv \leq \frac{4nc_1}{m}|M_{ij}|/\sqrt{\frac{|M_{ij}|}{\|M\|_F}} \leq \frac{4nc_1}{m} \sqrt{|M_{ij}| \|M\|_F}$. 

Hence applying Bernstein inequality with $t=\delta \sqrt{\|M_j\|^2 +|M_{ij}| \|M\|_F}$ gives, $\sum_i \dij \wij \ut_i M_{ij} \leq (1+\delta_1)\sqrt{\|M_j\|^2 + |M_{ij}| \|M\|_F}$ with probability greater than $1-\frac{2}{n^3}$ when $m \geq \frac{24 c_1^2}{\delta_1^2}n \log(n)$. For $\delta_1 \leq \frac{1}{20}$, we get, $\|\uht\| \lv\widehat{v}_j^{t+1}\rv \leq \frac{21}{19} \sqrt{\|M_j\|^2 + |M_{ij}| \|M\|_F}$.

Now we will bound $\|\widehat{v}^{t+1}\|$.
\begin{align}
\|\uht\|\|\widehat{v}^{t+1}\| &\geq  \|\uht\|\ip{\widehat{v}^{t+1}}{\vo} \stackrel{\zeta_1}{\geq} \so\ip{\ut}{\uo} - \so \frac{\delta_1}{1-\delta_1} \sqrt{1 -\ip{\uo}{\ut}^2}- \frac{1}{1-\delta_1}\| y^T \vo\|  \\
&\stackrel{\zeta_2}{\geq}  \so\ip{\widehat{u}^0}{\uo} - 2\so\delta_1\sqrt{1 -\ip{\uo}{\widehat{u}^0}^2} - 2\delta\|M-M_1\|\stackrel{\zeta_3}{\geq} \frac{2}{5}\so.
\end{align}
$\zeta_1$ follows from Lemma~\ref{lem:approx_supp3} and equations~\eqref{eq:witerates1} and~\eqref{eq:supp_b}. $\zeta_2$ follows from using $\ip{\uo}{\widehat{u}^0} \leq \ip{\uo}{\ut}$ and $\delta_1 \leq \frac{1}{20}$. $\zeta_3$ follows from the argument: for $\delta_1 \leq \frac{1}{16}$, $\ip{\widehat{u}^0}{\uo} - 2\delta_1\sqrt{1 -\ip{\uo}{\widehat{u}^0}^2} $ is greater than $\frac{1}{2}$, if $\ip{\widehat{u}^0}{\uo} \geq \frac{3}{5}$. This holds because $dist(\uo, \widehat{u}^0) \leq \frac{4}{5}$ from Lemma~\ref{lem:approx_init}. 

Hence we get \begin{align*} \vto_j = \frac{\wvto_j}{\|\wvto\|} \leq  3\frac{\sqrt{\|M_j\|^2 + |M_{ij}| \|M\|_F}}{\so}  \leq  c_1  \sqrt{ \|M_j\|^2/\|M\|_F^2 + |M_{ij}|/\|M\|_F} ,\end{align*} for $c_1 =6.$

Hence we have shown that $\vto$ satisfies the row norm bounds.  From Lemma~\ref{lem:noise_samplebound} we have, in each iteration with probability greater than $1-\frac{\gamma}{T \log(n)}$ we have $\lV (\ut)^T\Ro(M-M_1)  \rV \leq dist(\ut, \uo)\|M-M_1\| + \delta \lV M-M_1 \rV_F$. Hence the probability of failure in $T$ iterations is less than $\gamma.$ Lemma now follows from assumption on $m$.
\end{proof}

Now we have all the elements needed for proof of the Theorem~\ref{thm:main}.\\
{\bf Proof of Theorem~\ref{thm:main}:}[Rank-1 case]
 \begin{proof}
Lemma~\ref{lem:approx_init} has shown that $\widehat{u}^0$ satisfies the row norm bounds condition. From Lemma~\ref{lem:waltmin_descent} we get $dist(\vt, \vo) \leq  \frac{1}{2}dist(\ut, \uo) + 5\delta \|M-M_1\|_F/\so$. Hence $dist(\vt, \vo) \leq  \frac{1}{4^t}dist(\widehat{u}^0, \uo)  + 10\delta \|M-M_1\|_F/\so $. After $t= O(\log(\frac{1}{\zeta}))$ iterations we get $dist(\vt, \vo) \leq \zeta +  10 \delta \|M-M_1\|_F/\so$ and $dist(\ut, \uo) \leq \zeta +  10 \delta \|M-M_1\|_F/\so$. 

Hence, \begin{align*} \|M_1- \uht (\vht)^T\| &\leq  \|(I -\ut (\ut)^T) M_1\| + \| \ut \left((\ut)^TM_1 - (\vt)^T\right)\|\\ 
 &\stackrel{\zeta_1}{\leq} \so_1 dist(\ut, \uo) + \| \so B^{-1} (\ip{\ut}{\uo}B -C)\vo \| + \|B^{-1} y \| \\
 &\stackrel{\zeta_2}{\leq} \so_1 dist(\ut, \uo) +2\delta_3 \|M\| dist(\ut, \uo) + 2dist(\ut, \uo)\|M-M_1\|+ 2\delta\lV M-M_1 \rV_F \\
 &\leq  c\so_1 \zeta + \eps  \lV M-M_1 \rV_F .
 \end{align*}
 $\zeta_1$ follows from equation~\eqref{eq:witerates1} and $\zeta_2$ from $\|B^{-1}\| \leq \frac{1}{1-\delta_3} \leq 2$ from Lemma~\ref{lem:approx_supp2}.

 From Lemma~\ref{lem:noise_samplebound} we have, in each iteration with probability greater than $1-\frac{\gamma}{T \log(n)}$ we have $\lV (\ut)^T\Ro(M-M_1)  \rV \leq dist(\ut, \uo)\|M-M_1\| + \delta \lV M-M_1 \rV_F$. Hence the probability of failure in $T$ iterations is less than $\gamma.$
 \end{proof}

%% file: appendixb.tex
\subsection{Rank-$r$ proofs}\label{sec:main_proof_altmin}
Let SVD of $M_r$ be $\Uo \So (\Vo)^T$, $\Uo ,\Vo$ are $n \times r$ orthonormal matrices and $\So$ is a $r \times r$ diagonal matrix with $\So_{ii} =\so_i$. We have seen in Lemma~\ref{lem:approx_init} that initialization and trimming steps give  \begin{equation*}\|(\widehat{U}^{(0)})^i\| \leq 8\sqrt{r} \sqrt{ \|M^i\|^2/\|M\|_F^2} ~\text{ and }~ dist(\widehat{U}^{(0)}, \Uo) \leq \frac{1}{2}.\end{equation*} for $m \geq nr^3 \kappa^2 \log(n)$. 

In this section we will present rank-$r$ proof of Lemma~\ref{lem:waltmin_descent}. Before that we will present rank-$r$ version of the supporting lemmas.

Now like shown in~\cite{jain2013low}, we will analyze a equivalent algorithm to algorithm~\ref{algo:2} where the iterates are orthogonalized at each step. This makes analysis significantly simpler to present. Let, $\Uht =\Ut R^{(t)}$ and $\Vht =\Vt \Rt$ be the respective QR factorizations. Then we replace step 7 of the algorithm~\ref{algo:2} with \[\widehat{V}^{(t+1)}= \arg \min_{V\in \R^{n \times r} } \| R_{\Omega_{2t+1}}(M- \Ut V^T)\|_F^2.\] We similarly change step 8 too.  


We also assume that samples for each iteration are generated independently. For simplicity we will drop the subscripts on $\Omega$ that denote different set of samples in each iteration in the rest of the proof. The weighted alternating minimization updates at the $t+1$ iteration are, \begin{equation} \label{eq:witeratesr}(\Vht)^j =( B^j)^{-1}\left(C^j \So (\Vo)^j + (\Ut)^T \Ro(M-M_r)_j \right),\end{equation} where $B^j$ and $C^j$ are $r \times r$ matrices. \begin{equation*} B^j = \sum_i \dij \wij (\Ut)^i {(\Ut)^i}^T , ~~ C^j= \sum_i \dij \wij (\Ut)^i {(\Uo)^i}^T. \end{equation*} Writing in terms of power method updates we get,  \begin{equation}\label{eq:witeratesr1} (\Vht)^j  =\left((\Ut)^T \Uo-  (B^j)^{-1}(B^j(\Ut)^T \Uo - C^j )\right)\So (\Vo)^j + ( B^j)^{-1}(\Ut)^T \Ro(M-M_r)_j .\end{equation} Hence,  \begin{equation*} (\Vht)^T  =(\Ut)^T \Uo \So (\Vo)^T-  F + \sum_{j=1}^n ( B^j)^{-1}(\Ut)^T \Ro(M-M_r)_j e_j^T ,\end{equation*}   where the $j$th column of $F$, $F_j = \left((B^j)^{-1}(B^j(\Ut)^T \Uo - C^j )\right)\So (\Vo)^j $.

First we will bound $\|B^j\|$ using matrix Bernstein inequality.
\begin{lemma}\label{lem:approx_supp2r}
For $\Omega$ generated according to~\eqref{eq:prob} the following holds:
\begin{equation}  \|B^j -I\| \leq \delta_2 ,~~\text{and}~~  \|C^j- (\Ut)^T \Uo\| \leq \delta_2, \end{equation} with probability greater that $1-\frac{2}{n^{2}}$, for $m \geq \beta n r \kappa \log(n)$, $\beta \geq \frac{4*48 c_1^2}{\delta_2^2}$ and $\delta_2 \leq 3r$.
\end{lemma}

\begin{proof}
Let the matrices $X_i = \dij \wij (\Ut)^i {(\Ut)^i}^T$, then $ B^j = \sum_{i=1}^n X_i$. $\expec{X_i} =  (\Ut)^i {(\Ut)^i}^T$ and $\expec{B^j} =(\Ut)^T \Ut =I$. Also since $\Ut$ satisfies~\eqref{eq:infty_approxr}, it is easy to see that $\lV X_i-\expec{X_i} \rV \leq \frac{16 c_1^2 n}{m}$ and $\lV \expec{\sum_i (X_i-\expec{X_i}) (X_i-\expec{X_i})^T} \rV \leq \frac{16 c_1^2 r n}{m}$. Applying the matrix Bernstein inequality gives the first result.

Let the matrices $Y_i = \dij \wij (\Ut)^i {(\Uo)^i}^T$, then  $ C^j = \sum_{i=1}^n Y_i$.  $\expec{Y_i} =  (\Ut)^i {(\Uo)^i}^T$ and $\expec{C^j} =(\Ut)^T \Uo$. Also since $\Ut$ satisfies~\eqref{eq:infty_approxr}, it is easy to see that $\lV Y_i-\expec{Y_i} \rV \leq \frac{8 c_1 \kappa r^{0.5} n}{m}$ and $\lV \expec{\sum_i (Y_i-\expec{Y_i}) (Y_i-\expec{Y_i})^T} \rV \leq \frac{8 c_1^2 r n}{m}$. Applying the matrix Bernstein inequality gives the second result.
\end{proof}

Now we will bound the error caused by the $M-M_r$ component in each iteration. 
\begin{lemma}\label{lem:noise_sampleboundr}
For $\Omega$ generated according to~\eqref{eq:prob} the following holds:
\begin{equation}
\lV (\Ut)^T\Ro(M-M_r) - (\Ut)^T (M-M_r) \rV \leq \delta \|M-M_r\|_F,
\end{equation}
 with probability greater that $1-\frac{1}{c_2 \log(n)}$, for $m \geq \beta nr \log(n)$, $\beta \geq \frac{4 c_1^2 c_2}{ \delta^2}$. Hence,$\lV (\Ut)^T\Ro(M-M_r)  \rV \leq dist(\Ut, \Uo)\|M-M_r\| + \delta \lV M-M_r \rV_F,$ for constant $\delta$.
\end{lemma}
\begin{proof}[Proof of lemma~\ref{lem:noise_sampleboundr}]
Let the random matrices $X_{ij} = (\dij-\qhij)\wij (M-M_r)_{ij} (\Ut)^i e_j^T$. Then $\sum_{ij} X_{ij} =(\Ut)^T\Ro(M-M_r) - (\Ut)^T (M-M_r)$. Also $\expec{X_{ij} } =0$. We will use the matrix Chebyshev inequality for $p=2$. Now we will bound $\expec{\lV\sum_{ij}X_{ij}\rV_2^2}$. \begin{align*} \expec{\lV\sum_{ij}X_{ij}\rV_2^2} &= \expec{\sum_j \lV \left(\sum_i  (\dij-\qhij)\wij (M-M_r)_{ij} (\Ut)^i \right)^2 \rV} \\
&\stackrel{\zeta_1}{=} \sum_j \sum_i  \qhij(1-\qhij)(\wij)^2 (M-M_r)_{ij}^2 \|(\Ut)^i\|^2  \\
&\leq  \sum_{ij} \wij \|(\Ut)^i\|^2 (M-M_r)_{ij}^2 \\
&\stackrel{\zeta_2}{\leq} \frac{ c_1^2n}{m} \|M-M_1\|_F^2.
\end{align*}
$\zeta_1$ follows from the fact that $X_{ij}$ are zero mean independent random variables. $\zeta_2$ follows from~\eqref{eq:infty_approxr}. Hence applying the matrix Chebyshev inequality for $p=2$ and $t=\delta \|M-M_1\|_F$ gives the result.
\end{proof}

Since $\|B^j\|$ is bounded by the previous lemma, to get a bound on the norm of the error term in equation~\eqref{eq:witeratesr1}, $\|F\|$, we need to bound $\|\tilde{F}\|$, where $\tilde{F}_j = (B^j(\Ut)^T \Uo - C^j )\So (\Vo)^j$.
\begin{lemma}\label{lem:approx_supp3r}
For $\Omega$ generated according to~\eqref{eq:prob} the following holds:
\begin{equation} \|\tilde{F}\| \leq \delta_2 \so_1 dist(\Ut, \Uo) \end{equation} with probability greater that $1-\frac{2}{n^{2}}$, for $m \geq \beta n  \log(n)$, $\beta \geq \frac{32 r^3 \kappa^2 }{\delta_2^2}$, $c_1 \leq 8\kappa \sqrt{r}$ and $\delta_2 \leq \frac{1}{2}$.
\end{lemma}

\begin{proof}
Recall that the $j$th column of $F$, $F_j =\left((B^j)^{-1}(B^j(\Ut)^T \Uo - C^j )\right)\So (\Vo)^j$, where $B^j = \sum_i \dij \wij (\Ut)^i {(\Ut)^i}^T $, and $C^j= \sum_i \dij \wij (\Ut)^i {(\Uo)^i}^T.$ We will bound spectral norm of $F$ using matrix Bernstein inequality.

Let $u^i =(\Ut)^i$ and \[ A^j_i = u^i (u^i)^T (\Ut)^T \Uo - u^i {(\Uo)^i}^T .\] Let $X_{ij} =\left((B^j)^{-1}(\dij \wij A^j_i )\right)\So (\Vo)^j e_j^T$, then $\sum_{i} X_{ij} =\left((B^j)^{-1}(B^j(\Ut)^T \Uo - C^j )\right)\So (\Vo)^j e_j^T$.

Now we will bound $||X_{ij}||$.
\begin{align*}
\lV X_{ij} \rV &\leq \frac{\so_1}{1-\delta_2} \lV (\Vo)^j \rV \wij \lV A^j_i \rV \leq \frac{\so_1}{1-\delta_2} \lV (\Vo)^j \rV \wij \lV (\Ut)^i \rV \lV (u^i)^T (\Ut)^T \Uo - {(\Uo)^i}^T\rV \\
&\stackrel{\zeta_1}{\leq }\frac{\so_1}{1-\delta_2}  \frac{8c_1n \sqrt{r} \kappa}{m} dist(\Ut, \Uo).
\end{align*}
$\zeta_1$ follows from \eqref{eq:infty_approxr}, \eqref{eq:prop_a4} and $$\lV (u^i)^T (\Ut)^T \Uo - {(\Uo)^i}^T\rV  =\lV e_i^T\left(\Ut (\Ut)^T \Uo - (\Uo) \right)\rV  \leq \lV \Ut (\Ut)^T \Uo - (\Uo) \rV =dist(\Ut, \Uo).$$

Similarly let us bound the variance $\lV \expec{ \sum_{ij} X_{ij} X_{ij}^T} \rV$.
\begin{align*}
\lV \expec{ \sum_{ij} X_{ij} X_{ij}^T} \rV &= \lV \expec{ \sum_{ij} \dij \wij^2   (B^j)^{-1} A^j_i \So (\Vo)^j {(\Vo)^j}^T \So {A^j_i}^T ((B^j)^{-1})^T e_j e_j^T } \rV \\
&\leq \sum_{ij} \frac{1}{(1-\delta_2)^2} \lV \wij A^j_i \So (\Vo)^j {(\Vo)^j}^T \So {A^j_i}^T\rV \\
&\leq \sum_{ij}\frac{1}{(1-\delta_2)^2} (\so_1)^2 \wij \lV (\Vo)^j \rV^2 \lV (\Ut)^i \rV^2 \lV (u^i)^T (\Ut)^T \Uo - {(\Uo)^i}^T\rV ^2 \\
&\stackrel{\zeta_1}{\leq} \frac{(\so_1)^2}{(1-\delta_2)^2} \frac{8nc_1^2}{m} \sum_{ij}  \lV (\Vo)^j \rV^2 \lV (u^i)^T (\Ut)^T \Uo - {(\Uo)^i}^T\rV ^2 \\
&\stackrel{\zeta_2}{\leq} \frac{(\so_1)^2}{(1-\delta_2)^2} \frac{8nc_1^2}{m}  \lV \Ut (\Ut)^T \Uo - (\Uo) \rV_F^2 \sum_{j}  \lV (\Vo)^j \rV^2 \\
&\leq \frac{(\so_1)^2}{(1-\delta_2)^2} \frac{8nc_1^2}{m} r^2 dist(\Ut, \Uo)^2.
\end{align*}
$\zeta_1$ follows from \eqref{eq:infty_approxr} and \eqref{eq:prop_a4}. $\zeta_2$ follows from 

Similarly $\lV \expec{ \sum_{ij} X_{ij}^T X_{ij}} \rV$ can be bounded. Now applying the matrix Bernstein inequality with $t= \delta_2 \so_1 dist(\Ut, \Uo)$ gives the result.
\end{proof}

Now since $\Vht =\Vt \Rt$, \begin{align}\sigma_{\min}(\Rt) =\sigma_{\min}(\Vht) \stackrel{\zeta_1}{\geq} \sigma_{\min}((\Ut)^T \Uo \So (\Vo)^T) - \|F\|- \|\sum_{j=1}^n ( B^j)^{-1}(\Ut)^T \Ro(M-M_r)_j e_j^T\|. \end{align} $\zeta_1$ follows from~\eqref{eq:witeratesr1}. Now  $\sigma_{\min}((\Ut)^T \Uo \So (\Vo)^T) \geq \so_r \sqrt{1-dist(\Ut, \Uo)^2}$. $\|F\| \leq \frac{1}{1-\delta_2}\|\tilde{F}\| \leq  \frac{\delta_2}{1-\delta_2} \so_1 dist(\Ut, \Uo)$. 

\begin{align*} \|\sum_{j=1}^n ( B^j)^{-1}(\Ut)^T \Ro(M-M_r)_j e_j^T\| &\leq  \frac{1}{1-\delta_2}\|(\Ut)^T \Ro(M-M_r)\|_F \\&\leq \frac{1}{1-\delta_2}dist(\Ut, \Uo)\|M-M_r\| + \frac{\delta}{1-\delta_2} \lV M-M_r\rV_F , \end{align*} from Lemma~\ref{lem:noise_sampleboundr}.

Hence \begin{align*}\sigma_{\min}(\Rt)  &\geq \so_r \left(\sqrt{1-dist(\Ut, \Uo)^2}-\kappa\frac{\delta_2}{1-\delta_2}dist(\Ut, \Uo) - \frac{2}{1-\delta_2} \lV M-M_r\rV_F/\so_r \right) \\
& \geq \frac{\so_r}{2},\end{align*} for enough number of samples $m$.

Now we are ready to present proof of Lemma~\ref{lem:waltmin_descent} for rank-$r$ case.\\
\noindent {\bf Proof of Lemma~\ref{lem:waltmin_descent}:}
\begin{proof}
The proof like in rank-$1$ case has two steps. In the first step we show that $dist(\Vt, \Vo)$ decreases in each iteration. In the second step we show row norm bounds for $\Vt$. Recall from the assumptions of the lemma we have the following row norm bound for $\Ut$:\begin{equation}\label{eq:infty_approxr}\|(\Ut)^i\| \leq c_1\sqrt{ \|M^i\|^2/\|M\|_F^2 + |M_{ij}|/\|M\|_F} , ~~\text{ for all } i. \end{equation}\\

\noindent {\bf Bounding $dist(\Vt, \Vo)$:}
\begin{align*}
dist(\Vt, \Vo) &=\|(\Vt)^T \Vo_{\perp}\| \stackrel{\zeta_1}{\leq} \|{(\Rt)^{-1}}^T F \Vo_{\perp}\| + \frac{1}{1-\delta_2}\|{(\Rt)^{-1}}^T(\Ut)^T \Ro(M-M_r) \Vo_{\perp}\|_F \\
&\stackrel{\zeta_2}{\leq}  \frac{1}{\sigma_{\min}(\Rt)}\left( \| F\| + \frac{1}{1-\delta_2}dist(\Ut, \Uo)\|M-M_r\| + \frac{\delta}{1-\delta_2} \lV M-M_r\rV_F\right) \\
&\leq  \frac{2}{ \so_r } \left(\frac{\delta_2}{1-\delta_2} \|M\| dist(\Ut, \Uo) +  \frac{1}{1-\delta_2}dist(\Ut, \Uo)\|M-M_r\| + \frac{\delta}{1-\delta_2} \lV M-M_r\rV_F \right)\\
&\leq \frac{1}{2}dist(\Ut, \Uo) + 5\delta \lV M-M_r\rV_F/\so_r,
\end{align*}
 for $\delta_2 \leq \frac{1}{16 \kappa}$. $\zeta_1$ follows from~\eqref{eq:witeratesr1}. $\zeta_2$ follows from Lemma~\ref{lem:noise_sampleboundr}.\\

\noindent {\bf Bounding $\|(\Vt)^j\|$:}

From Lemma~\ref{lem:approx_supp2r} and~\eqref{eq:infty_approxr} we get that $\sigma_{\min}(B^j) \geq 1-\delta_2$ and $\sigma_{\max}(C^j) \leq 1+ \delta_2$. Recall that  \begin{equation}  (\Vt)^j  ={(\Rt)^{-1}}^T \left(( B^j)^{-1}(\Ut)^T \Ro(M)_j \right) .\end{equation} Hence, $ \| (\Vt)^j \| \leq \frac{1}{\sigma_{\min}(\Rt)}\left(  \frac{1}{1-\delta_2}\|(\Ut)^T \Ro(M)_j \|\right).$ We will bound $\|(\Ut)^T \Ro(M)_j \|$ using matrix Bernstein inequality.

Let $X_{i} =(\dij-\qhij) \wij M_{ij} (\Ut)^i e_j^T$. Then $\expec{X_i} =0$ and $\sum_{i}X_{i} = (\Ut)^T \Ro(M)_j -(\Ut)^T M_j$. Now $\|X_i\| \leq  \frac{c_1 2n}{m} \sqrt{|M_{ij}| \|M\|_F}$ and $\lV \expec{\sum_i X_i X_i^T} \rV \leq \frac{8 c_1^2 n}{m} \|M_j\|^2.$ Hence applying matrix Bernstein inequality with $t=\delta_2 \sqrt{ \|M_j\|^2 + |M_{ij}| \|M\|_F}$, implies $$\|(\Ut)^T \Ro(M)_j \| \leq \|M_j\| + \delta_2 \sqrt{ \|M_j\|^2 + |M_{ij}| \|M\|_F} \leq (1+\delta_2)\sqrt{ \|M_j\|^2 + |M_{ij}| \|M\|_F}$$ with probability greater than $1-\frac{2}{n^2}$ for $m \geq \frac{24 c_1^2}{\delta_2^2}n \log(n)$. Hence $ \| (\Vt)^j \| \leq  8\kappa \sqrt{r}\sqrt{ \frac{\|M_j\|^2}{\|M\|_F^2} + \frac{|M_{ij}|}{ \|M\|_F}}.$

Hence we have shown that $(\Vt)^j$ satisfies corresponding row norm bound. This completes the proof of the Lemma.
\end{proof}

Now we have all the elements needed for proof of the Theorem~\ref{thm:main}.\\
{\bf Proof of Theorem~\ref{thm:main}:}
 \begin{proof}
 From Lemma~\ref{lem:waltmin_descent} we get $dist(\Vt, \Vo) \leq  \frac{1}{2}dist(\Ut, \Uo) + 5\delta \|M-M_r\|_F/\so_r$. Hence $dist(\Vt, \Vo) \leq  \frac{1}{4^t}dist(\widehat{U}^0, \Uo)  + 10\delta \|M-M_r\|_F/\so_r $. After $t= O(\log(\frac{1}{\zeta}))$ iterations we get $dist(\Vt, \Vo) \leq \zeta +  10 \delta \|M-M_r\|_F/\so_r$ and $dist(\Ut, \Uo) \leq \zeta +  10 \delta \|M-M_r\|_F/\so_r$. 

Hence, \begin{align*} \|M_r- \Ut (\Vht)^T\| &\leq  \|(I -\Ut (\Ut)^T) M_r\| + \| \Ut \left((\Ut)^TM_r - (\Vht)^T\right)\|\\ 
 &\stackrel{\zeta_1}{\leq} \so_1 dist(\Ut, \Uo) + \|  F \| + \|\sum_{j=1}^n ( B^j)^{-1}(\Ut)^T \Ro(M-M_r)_j e_j^T\| \\
 &\stackrel{\zeta_2}{\leq} \so_1 dist(\Ut, \Uo) +2\delta_2 \so_1 dist(\Ut, \Uo) + 2dist(\Ut, \Uo)\|M-M_r\|+ 2\delta\lV M-M_r \rV_F \\
 &\leq  c\so_1 \zeta + \eps  \lV M-M_1 \rV_F .
 \end{align*}
 $\zeta_1$ follows from equation~\eqref{eq:witeratesr1} and $\zeta_2$ from $\|B^{-1}\| \leq \frac{1}{1-\delta_3} \leq 2$ from Lemma~\ref{lem:approx_supp2r}.

 From Lemma~\ref{lem:noise_sampleboundr} we have, in each iteration with probability greater than $1-\frac{\gamma}{T \log(n)}$ we have $\lV (\Ut)^T\Ro(M-M_r)  \rV \leq dist(\Ut, \Uo)\|M-M_r\| + \delta \lV M-M_r \rV_F$. Hence the probability of failure in $T$ iterations is less than $\gamma.$
 \end{proof}

%% file: appendixc.tex
\section{Proofs of section~\ref{sec:covariance}}\label{sec:pfcov}

We will now discuss proof of Theorem~\ref{thm:mult}. The proof follows same structure as proof of Theorem~\ref{thm:main} with few key changes because of the absence of $L1$ term in the sampling and the special structure of $M=AB$. Again for simplicity we will present proofs only for the case of $n_1 =n_2 =n$. 

Recall that $\qohij=\min(1, q_{ij})$ where $q_{ij} =m\cdot\left(\frac{\|A^i\|^2}{n\|A\|_F^2} + \frac{\|B_j\|^2}{n\|B\|_F^2}\right)$. Also, let $\wij =1/\qohij$. 

First we will abstract out the properties of the sampling distribution~\eqref{eq:prob_mp} that we use in the rest of the proof. Also let $C_{AB}=\frac{(\|A\|_F^2 +\|B\|_F^2)^2}{\|AB\|_F^2}$

\begin{lemma}\label{lem:mult_supp1}
For $\Omega$ generated according to~\eqref{eq:prob_mp}  and under the assumptions of Lemma~\ref{lem:prod_init} the following holds, for all $(i ,j)$ such that $  \qoij \leq 1$.
\begin{equation}\label{eq:prop_m1} \frac{M_{ij}}{\qohij} \leq \frac{n}{2m}(\|A\|_F^2 +\|B\|_F^2),\end{equation} 
\begin{equation}\label{eq:prop_m2} \sum_{ \{j:  \qohij =\qoij \}} \frac{M_{ij}^2}{\qohij} \leq \frac{n}{m}(\|A\|_F^2 +\|B\|_F^2)^2, \end{equation}
\begin{equation}\label{eq:prop_m3} \frac{\|(\Uo)^i\|^2}{\qohij} \leq \frac{n}{m}\frac{(\|A\|_F^2 +\|B\|_F^2)^2}{\|A\cdot B\|_F^2}, \end{equation} and
\begin{equation}\label{eq:prop_m4} \frac{\|(\Uo)^i\| \|(\Vo)^j\|}{\qohij}  \leq  \frac{n}{m}\frac{(\|A\|_F^2 +\|B\|_F^2)^2}{\|A\cdot B\|_F^2}. \end{equation}  
\end{lemma}

The proof of the lemma~\ref{lem:mult_supp1} is straightforward from the definition of $\qij$.

Now, similar to proof of Theorem~\ref{thm:main}, we divide our analysis in two parts: initialization analysis and weighted alternating minimization analysis. 

\subsection{Initialization}
\begin{lemma}[Initialization]\label{lem:prod_init}
Let the set of entries $\Omega$ be generated according to $\qohij$~\eqref{eq:prob_mp}. Also, let $m \geq C \cab \frac{n}{\delta^2} \log(n)$. Then, the following holds (w.p. $\geq 1-\frac{2}{n^{10}}$): 
\begin{equation}
\lV \Ro(AB) - AB \rV \leq \delta \lV AB \rV_F.
\end{equation}
Also, if $\|AB-(AB)_r\|_F \leq \frac{1}{576\kappa r^{1.5}}\|(AB)_r\|_F$, then the following holds (w.p. $\geq 1-\frac{2}{n^{10}}$): 
$$\|(\widehat{U}^{(0)})^i\| \leq 8\sqrt{r} \sqrt{ \|A^i\|^2/\|A\|_F^2} ~\text{ and }~ dist(\widehat{U}^{(0)}, \Uo) \leq \frac{1}{2},$$
where $\widehat{U}^{(0)}$ is the initial iterate obtained using Steps 4, 5 of Sub-Procedure~\ref{algo:2}. $\kappa=\sigma_1^*/\sigma_r^*$, $\sigma_i^*$ is the $i$-th singular value of $AB$, $(AB)_r=\Uo\Sigma^*(\Vo)^T$. 
\end{lemma}

\begin{proof}
First we show that $\Ro(AB)$ is a good approximation of $AB$. 

Let $M=AB.$ We prove this part of the lemma using the matrix Bernstein inequality. Let $X_{ij} = (\dij -\qohij) \wij M_{ij} e_i e_j^T$. Note that $\{ X_{ij} \}_{i,j=1}^n$ are independent zero mean random matrices. Also $\Ro(AB) -\expec{\Ro(AB)} = \sum_{ij}X_{ij}$.

First we will bound $\|X_{ij}\|$. When $m\qoij \geq 1$, $\qohij=1$ and $\dij=1$, and $X_{ij} =0$ with probability 1. Hence we only need to consider cases when $\qohij= m\qoij \leq 1$. We will assume this in all the proofs without explicitly mentioning it any more.
\begin{align*} \|X_{ij}\| = \max \{ \lv(1- \qohij) \wij M_{ij} \rv, \lv \qohij \wij M_{ij}\rv \}. \end{align*}  Recall $\wij =1/\qohij$. Hence
\begin{align*}
  \lv(1- \qohij) \wij M_{ij}\rv  &=  \lv(\frac{1}{\qohij}- 1)  M_{ij}\rv  \leq \lv \frac{M_{ij}}{\qohij} \rv  \stackrel{\zeta_1}{\leq} \frac{n}{2m}(\|A\|_F^2 +\|B\|_F^2).
\end{align*}
$\zeta_1$ follows from~\eqref{eq:prop_m1}.
\begin{align*}
 \lv \qohij \wij M_{ij}\rv  = \lv M_{ij} \rv \stackrel{\zeta_1}{\leq} \lv\frac{M_{ij}}{\qohij} \rv \leq \frac{n}{2m}(\|A\|_F^2 +\|B\|_F^2).
\end{align*}
$\zeta_1$ follows from $\qohij \leq 1$.

Hence, $\|X_{ij}\| $ is bounded by $L = \frac{n}{2m}(\|A\|_F^2 +\|B\|_F^2)$. Recall that this is the step in the proof of Lemma~\ref{lem:approx_init} that required the L1 term in sampling, which we didn\rq{}t need now because of the structure $AB$ of the matrix. Now we will bound the variance.

\begin{align*}
\lV \expec{ \sum_{ij} X_{ij} X_{ij}^T} \rV &= \lV \expec{ \sum_{ij} (\dij -\qohij)^2 \wij^2 M_{ij}^2 e_i e_i^T} \rV =\lV  \sum_{ij}  \qohij(1- \qohij) \wij^2 M_{ij}^2 e_i e_i^T \rV \\
&= \max_i \lv  \sum_{j} \qohij(1- \qohij) \wij^2 M_{ij}^2  \rv.
\end{align*}
Now,
\begin{align*}
 \sum_{j} \qohij(1- \qohij) \wij^2 M_{ij}^2 =\sum_j (\frac{1}{\qohij}-1) M_{ij}^2 \leq \sum_j \frac{ M_{ij}^2}{(\qohij)}  \stackrel{\zeta_1}{\leq}  \frac{n}{m}(\|A\|_F^2 +\|B\|_F^2)^2.
\end{align*}
$\zeta_1$ follows from~\eqref{eq:prop_m2}.
Hence \begin{align*} \lV \expec{ \sum_{ij} X_{ij} X_{ij}^T} \rV = \max_i \lv  \sum_{j} \qohij(1- \qohij) \wij^2 M_{ij}^2  \rv \leq \max_i   \frac{n}{m}(\|A\|_F^2 +\|B\|_F^2)^2 =  \frac{n}{m}(\|A\|_F^2 +\|B\|_F^2)^2. \end{align*} We can prove the same bound for the $\lV \expec{ \sum_{ij} X_{ij}^T X_{ij}} \rV$. Hence $\sigma^2 =  \frac{n}{m}(\|A\|_F^2 +\|B\|_F^2)^2$. Now using matrix Bernstein inequality with $t=\delta \|AB\|_F$ gives, with probability $ \geq 1-\frac{2}{n^{2}}$,
\begin{equation*}
\lV \Ro(AB) -\expec{\Ro(AB)}\rV =\lV \Ro(AB) -AB\rV \leq  \delta \lV AB \rV_F.
\end{equation*}

Once we have $\|\Ro(M) -M\| \leq \delta \|M\|_F$,  proof of the trimming step that guarantees \begin{equation*}\|(\widehat{U}^{(0)})^i\| \leq 8\sqrt{r} \sqrt{ \|A^i\|^2/\|A\|_F^2} ~\text{ and }~ dist(\widehat{U}^{(0)}, \Uo) \leq \frac{1}{2},\end{equation*} follows from the same argument as in Lemma~\ref{lem:approx_init}.

\end{proof}

\subsection{Weighted AltMin Analysis}

\begin{lemma}[WAltMin Descent]\label{lem:prod_waltmin_descent}
Let hypotheses of Theorem~\ref{thm:mult} hold. Also, let $\|AB-(AB)_r\|_F \leq \frac{1}{576\kappa r\sqrt{r}}\|(AB)_r\|_F$. Let $\widehat{U}^{(t)}$ be the $t$-th step iterate of Sub-Procedure~\ref{algo:2} (called from $WAltMin(P_{\Omega}(A\cdot B), \Omega, \hat{q}, T)$), and let $\widehat{V}^{(t+1)}$ be the $(t+1)$-th iterate (for $V$). Also, let $\|(\Ut)^i\| \leq 8\sqrt{r}\kappa \sqrt{ \|A^i\|^2/\|A\|_F^2}$ and $dist({U}^{(t)}, \Uo) \leq \frac{1}{2}$, where $U^{(t)}$ is a set of orthonormal vectors spanning $\widehat{U}^{(t)}$. Then, the following holds (w.p. $\geq 1-\gamma/T$): 
$$dist({V}^{(t+1)}, V^*)\leq \frac{1}{2}dist({U}^{(t)}, \Uo)+ \epsilon \|AB-(AB)_r\|_F/\so_r,$$
and $\|(\Vt)^j\| \leq 8\sqrt{r}\kappa \sqrt{ \|B_j\|^2/\|B\|_F^2}$, where $V^{(t+1)}$ is a set of orthonormal vectors spanning $\widehat{V}^{(t+1)}$.
\end{lemma}

For the sake of simplicity we will discuss the proof for rank-1$(r=1)$ case for this part of the algorithm. Rank-$r$ proof follows by combining the below analysis with rank-$r$ analysis of Lemma~\ref{lem:waltmin_descent} (see Section~\ref{sec:main_proof_altmin}). Before presenting the proof of this Lemma, we will state couple of supporting lemmas. The proofs of these supporting lemmas follows very closely to the ones in section~\ref{sec:proofs_waltmin_descent}.

\begin{lemma}\label{lem:mult_supp2}
For $\Omega$ sampled according to~\eqref{eq:prob_mp}  and under the assumptions of Lemma~\ref{lem:prod_waltmin_descent}, the following holds: 
\begin{equation}
\lv \sum_j \dij \wij (\uo_j)^2- \sum_j (\uo_j)^2 \rv \leq \delta_1,
\end{equation}
 with probability greater that $1-\frac{2}{n^{2}}$, for $m \geq \beta \cab n \log(n)$, $\beta \geq \frac{16}{\delta_1^2}$ and $\delta_1 \leq 3$.
\end{lemma}

We assume that samples for each iteration are generated independently. For simplicity we will drop the subscripts on $\Omega$ that denote different set of samples in each iteration in the rest of the proof. The weighted alternating minimization updates at the $t+1$ iteration are, \begin{equation} \label{eq:witerates_cov}\|\uht\| \widehat{v}^{t+1}_j = \so\vo_j \frac{\sum_i \dij \wij \ut_i \uo_i }{\sum_i \dij \wij (\ut_i)^2} +\frac{\sum_i \dij \wij \ut_i (M-M_1)_{ij}}{\sum_i \dij \wij (\ut_i)^2}.\end{equation} Writing in terms of power method updates we get,  \begin{equation}\label{eq:witerates1_cov} \|\uht\| \widehat{v}^{t+1} =\so\ip{\uo}{\ut}\vo- \so P^{-1} (\ip{\ut}{\uo}P -Q)\vo  +P^{-1} y,\end{equation} where $P$ and  $Q$  are diagonal matrices with $P_{jj} = \sum_i  \dij \wij (\ut_i)^2$ and  $Q_{jj} =\sum_i \dij \wij \ut_i \uo_i$ and $y$ is the vector $\Ro(M-M_1)^T \ut$ with entries $y_{j}=\sum_i \dij \wij \ut_i (M-M_1)_{ij}$. 

Now we will bound the error caused by the $M-M_r$ component in each iteration. 
\begin{lemma}\label{lem:mult_noisebound}
For $\Omega$ generated according to~\eqref{eq:prob_mp} and under the assumptions of Lemma~\ref{lem:prod_waltmin_descent}, the following holds:
\begin{equation}
\lV (\Ut)^T\Ro(M-M_r) - (\Ut)^T (M-M_r) \rV \leq \delta \|M-M_r\|_F,
\end{equation}
 with probability greater that $1-\frac{1}{c_2 \log(n)}$, for $m \geq \beta nr \log(n)$, $\beta \geq \frac{4 c_1^2 c_2}{\delta^2}$. Hence, $\lV (\Ut)^T\Ro(M-M_r)  \rV \leq dist(\Ut, \Uo)\|M-M_r\| + \delta \lV M-M_r \rV_F,$ for constant $\delta$.
\end{lemma}

\begin{lemma}\label{lem:mult_supp3}
For $\Omega$ sampled according to~\eqref{eq:prob_mp} and under the assumptions of Lemma~\ref{lem:prod_waltmin_descent}, the following holds: 
\begin{equation}
\|(\ip{\ut}{\uo}P -Q)\vo\| \leq \delta_1 \sqrt{1 -\ip{\uo}{\ut}^2},
\end{equation}
with probability greater than $1-\frac{2}{n^2}$,  for $m \geq \beta \cab n\log(n), \beta \geq \frac{48 c_1^2 }{\delta_1^2}$  and $\delta_1 \leq 3$.
\end{lemma}

Now we will provide proof of lemma~\ref{lem:prod_waltmin_descent}.\\
{\bf Proof of lemma~\ref{lem:prod_waltmin_descent}:}[Rank-1 case]
\begin{proof}

Let $\ut$ and $\vt$ be the normalized vectors of the iterates $\uht$ and $\vht$. In the first step we will prove that the distance between $\ut$, $\uo$ and $\vt, \vo$ decreases with each iteration. In the second step we will prove that $v^{t+1}$ satisfies $|\vt_j| \leq c_1\sqrt{ \|B_j\|^2/\|B\|_F^2}$.  From the assumptions of the lemma we have, \begin{equation}\label{eq:infty_mult}|\ut_i| \leq c_1\sqrt{ \|A^i\|^2/\|A\|_F^2} .\end{equation} 

\noindent {\bf Bounding $\ip{\vto}{\vo}$:}

Using Lemma~\ref{lem:mult_supp2}, Lemma~\ref{lem:mult_supp3} and equation~\eqref{eq:witerates1_cov} we get,\begin{align}  \|\uht\|\ip{\widehat{v}^{t+1}}{\vo}  \geq \so\ip{\ut}{\uo} - \so \frac{\delta_1}{1-\delta_1} \sqrt{1 -\ip{\uo}{\ut}^2} - \frac{1}{1-\delta_1}\| y^T \vo\|\end{align} and \begin{align}  \|\uht\|\ip{\widehat{v}^{t+1}}{\vo_{\perp}}  \leq \so \frac{\delta_1}{1-\delta_1} \sqrt{1 -\ip{\uo}{\ut}^2} + \frac{1}{1-\delta_1}\|y \|.\end{align} Hence by applying the noise bounds Lemma~\ref{lem:mult_noisebound} we get,
\begin{align*}
dist(\vto, \vo)^2 &=1-\ip{\vto}{\vo}^2 = \frac{\ip{\widehat{v}^{t+1}}{\vo_{\perp}}^2}{\ip{\widehat{v}^{t+1}}{\vo_{\perp}}^2 + \ip{\widehat{v}^{t+1}}{\vo}^2} \leq  \frac{\ip{\widehat{v}^{t+1}}{\vo_{\perp}}^2}{ \ip{\widehat{v}^{t+1}}{\vo}^2} \\
&\stackrel{\zeta_1}{\leq} \frac{ 4(\delta_1 dist(\ut, \uo) +dist(\ut, \uo) \|M-M_1\|/\so + \delta \|M-M_1\|_F/\so)^2}{ (\ip{\ut}{\uo} -  2\delta_1 \sqrt{1 -\ip{\uo}{\ut}^2} - 2\delta\|M-M_1\|/\so)^2}\\
&\stackrel{\zeta_2}{\leq}\frac{ 4(\delta_1 dist(\ut, \uo) + dist(\ut, \uo)\|M-M_1\|/\so + \delta \|M-M_1\|_F/\so)^2}{ (\ip{\uo}{u^0} -  2\delta_1 \sqrt{1 -\ip{\uo}{u^0}^2}- 2\delta\|M-M_1\|/\so)^2} \\
&\stackrel{\zeta_3}{\leq} 25(\delta_1 dist(\ut, \uo) + dist(\ut, \uo)\|M-M_1\|/\so + \delta \|M-M_1\|_F/\so)^2.
\end{align*}
$\zeta_1$ follows from $\delta_1 \leq \frac{1}{2}$. $\zeta_2$ follows from using $\ip{\ut}{\uo} \geq \ip{u^0}{\uo}$. $\zeta_3$ follows from $(\ip{\uo}{u^0} -  2\delta_1 \sqrt{1 -\ip{\uo}{u^0}^2} \geq \frac{1}{2}$, $\delta \leq \frac{1}{20}$  and $\delta_1 \leq \frac{1}{20}$. Hence \begin{align} dist(\vto, \vo) &\leq \frac{1}{4}dist(\ut, \uo) +5dist(\ut, \uo)\|M-M_1\|/\so + 5\delta \|M-M_1\|_F/\so  \nonumber\\ &\leq \frac{1}{2}dist(\ut, \uo)  + 5\delta \|M-M_1\|_F/\so.\end{align}

Now, by selecting $m\geq \frac{C}{\gamma}\cdot \frac{(\|A\|_F^2+\|B\|_F^2)^2}{\|AB\|_F^2}\cdot \frac{n r^3}{(\eps)^2} \kappa^2 \log(n) \log^2(\frac{\|A\|_F+\|B\|_F}{\zeta})$, the above bound reduces to (w.p. $\geq 1-\gamma/\log(\frac{\|A\|_F+\|B\|_F}{\zeta})$): 
\begin{align} dist(\vto, \vo) \leq \frac{1}{2}dist(\ut, \uo)  + \epsilon \|M-M_1\|_F.\end{align}
Hence, using induction, after $T=\log(\frac{\|A\|_F+\|B\|_F}{\zeta})$ rounds, we obtain (w.p. $\geq 1-\gamma$): $dist(\vto, \vo)\leq \epsilon \|M-M_1\|_F+\zeta$. However, the above induction step would require $\vto$ to satisfy the $L_\infty$ condition as well, that we prove below. 
\vskip 0.2in
\noindent {\bf Bounding $\vto_j$:}

From Lemma~\ref{lem:mult_supp2} and~\eqref{eq:infty_mult} we get that $\lv \sum_i  \dij \wij (\ut_i)^2 - 1\rv \leq \delta_1$ and $\lv \sum_i  \dij \wij \uo_i \ut_i - \ip{\uo}{\ut}\rv \leq  \delta_1$, when $\beta \geq \frac{16 c_1^2}{\delta_1^2}.$ Hence, \begin{equation}\label{eq:supp_m1}1-\delta_1 \leq P_{jj} = \sum_i  \dij \wij (\ut_i)^2 \leq 1 +\delta_1, \end{equation} and  \begin{equation}\label{eq:supp_m2}Q_{jj} =\sum_i \dij \wij \ut_i \uo_i \leq \ip{\ut}{\uo} +\delta_1. \end{equation}
Recall that  \begin{equation*}  \|\uht\| \lv\widehat{v}_j^{t+1}\rv = \lv\frac{\sum_i \dij \wij \ut_i M_{ij} }{\sum_i \dij \wij (\ut_i)^2}\rv \leq  \frac{1}{1-\delta_1}\sum_i \dij \wij \ut_i M_{ij}.\end{equation*} 

We will bound using $\sum_i \dij \wij \ut_i M_{ij}$ using bernstein inequality. Let $X_i =(\dij-\qohij) \wij \ut_i M_{ij}$. Then $\sum_i \expec{X_i} =0 $ and $\sum_i \ut_i M_{ij} \leq \|M_j\|$ by Cauchy-Schwartz inequality. $\sum_i \Var(X_i) = \sum_i \qohij (1-\qohij) (\wij)^2 (\ut_i)^2 M_{ij}^2 \leq \sum_i \wij (\ut_i)^2 M_{ij}^2 \leq \frac{nc_1^2}{m} \|M_j\|^2 \leq \frac{nc_1^2}{m} \|B_j\|^2 \|A\|_F^2$. Finally $|X_{ij}| \leq \lv \wij \ut_i M_{ij} \rv \leq \frac{nc_1}{m}\|A\|_F \|B_j\|$. Hence applying bernstein inequality with $t=\delta \frac{\|B_j\|}{\|B\|_F}\|AB\|_F$ gives, $\sum_i \dij \wij \ut_i M_{ij} \leq (1+\delta_1) \frac{\|B_j\|}{\|B\|_F}\|AB\|_F$ with probability greater than $1-\frac{2}{n^3}$ when $m \geq \frac{24 c_1^2}{\delta_1^2}\cab n \log(n)$. For $\delta_1 \leq \frac{1}{20}$, we get, $\|\uht\| \lv\widehat{v}_j^{t+1}\rv \leq \frac{21}{19}  \frac{\|B_j\|}{\|B\|_F}\|AB\|_F$.

Now we will bound $\|\widehat{v}^{t+1}\|$.
\begin{align*}
\|\uht\|\|\widehat{v}^{t+1}\| &\geq  \|\uht\|\ip{\widehat{v}^{t+1}}{\vo} \stackrel{\zeta_1}{\geq} \so\ip{\ut}{\uo} - \so \frac{\delta_1}{1-\delta_1} \sqrt{1 -\ip{\uo}{\ut}^2}- \frac{1}{1-\delta_1}\| y^T \vo\|  \\
&\stackrel{\zeta_2}{\geq}  \so\ip{u^0}{\uo} - 2\so\delta_1\sqrt{1 -\ip{\uo}{u^0}^2} - 2\delta\|M-M_1\|\stackrel{\zeta_3}{\geq} \frac{2}{5}\so.
\end{align*}
$\zeta_1$ follows from Lemma~\ref{lem:mult_supp3} and equations~\eqref{eq:witerates1_cov} and~\eqref{eq:supp_m1}. $\zeta_2$ follows from using $\ip{\uo}{u^0} \leq \ip{\uo}{\ut}$ and $\delta_1 \leq \frac{1}{20}$. $\zeta_3$ follows by initialization and using the assumption on $m$ with large enough $C>0$. Hence we get \begin{align*} \vto_j = \frac{\wvto_j}{\|\wvto\|} \leq  3\frac{21}{19}   \frac{\|B_j\|}{\|B\|_F}\frac{\|AB\|_F}{\so} \leq  c_1 \frac{\|B_j\|}{\|B\|_F} ,\end{align*} for $c_1 =6.$

Hence we have shown that $\vto$ satisfies the row norm bounds. This completes the proof.
\end{proof}

The proof of the Theorem~\ref{thm:mult} now follows from the Lemma~\ref{lem:prod_init} and Lemma~\ref{lem:prod_waltmin_descent}.